\documentclass[onecolumn,12pt,draftcls]{IEEEtran}
\usepackage{amsmath}
\usepackage{amsfonts}
\usepackage{bbding}
\usepackage{amssymb}
\usepackage{array}
\usepackage{subfigure}

\usepackage{graphicx}
\usepackage{subfigure}
\usepackage[named]{algo}
\usepackage{algorithmic}
\usepackage{psfrag}
\usepackage{stfloats}
\usepackage[compress]{cite}
\makeatletter
\renewcommand{\citepunct}{,\penalty\@m\hskip.13emplus.1emminus.1em}
\renewcommand{\citedash}{\hbox{--}\penalty\@m}
\makeatother
\usepackage{setspace}
\usepackage{color}
\allowdisplaybreaks

\usepackage{amsthm}
\usepackage{stfloats}

\newtheorem{lemma}{Lemma}
\newtheorem{proposition}{Proposition}

\newtheorem{remark}{Remark}

\usepackage{hyperref}
\usepackage{bm}
\usepackage{bbm}
\renewcommand{\IEEEQED}{\IEEEQEDopen}

\IEEEoverridecommandlockouts

\begin{document}

\title{Prediction and Communication Co-design for Ultra-Reliable and Low-Latency Communications}

\author{
\IEEEauthorblockN{
Zhanwei Hou,
Changyang She,
Yonghui Li,
Zhuo Li,
and Branka Vucetic%,~\IEEEmembership{Fellow,~IEEE}
}\\
\thanks{Part of the work has been presented in IEEE international communications conference (ICC) 2019 \cite{hou2018ultra}.}
\thanks{Z. Hou, C. She, Y. Li and B. Vucetic are with the School of Electrical and Information Engineering, University of Sydney, Sydney, NSW 2006, Australia (email: \{zhanwei.hou, changyang.she, yonghui.li, branka.vucetic\}@sydney.edu.au). Z. Li is with  Beijing University of Technology, Beijing, China (email: zhuoli@bjut.edu.cn).}
}% <-this % stops a space
\maketitle

\begin{abstract}
Ultra-reliable and low-latency communications (URLLC) are considered as one of three new application scenarios in the fifth generation cellular networks. In this work, we aim to reduce the user experienced delay through prediction and communication co-design, where each mobile device predicts its future states and sends them to a data center in advance. Since predictions are not error-free, we consider prediction errors and packet losses in communications when evaluating the reliability of the system. Then, we formulate an optimization problem that maximizes the number of URLLC services supported by the system by optimizing time and frequency resources and the prediction horizon. Simulation results verify the effectiveness of the proposed method, and show that the tradeoff between user experienced delay and reliability can be improved significantly via prediction and communication co-design. Furthermore, we carried out an experiment on the remote control in a virtual factory, and validated our concept on prediction and communication co-design with the practical mobility data generated by a real tactile device.
\end{abstract}
\begin{IEEEkeywords}
Ultra-reliable and low-latency communications, prediction and communication co-design, delay-reliability tradeoff
\end{IEEEkeywords}

%Ensuring ultra-reliable and low-latency communications (URLLC) with limited bandwidth is crucial for 5G. The requirements of ultra-reliable and low-latency are challenge to meet simultaneously. As a fact, the reduction of delay is always at a cost of compromised reliability or increased bandwidth usage, which is limited by fundamental trade-offs among bandwidth-reliability-delay. Moreover, the end-to-end (E2E) communication delay is lower bounded by the propagation speed of electromagenetic wave. It is very challenging to improve the trade-off and even impossible to break through delay lower bound imposed by the physical limits by only designing and optimizing communication systems. As such,

\section{Introductions}
\subsection{Backgrounds and Motivations}
Ultra-reliable and low-latency communications (URLLC) are one of the new application scenarios in 5G communications \cite{3GPP2017Scenarios}. By achieving ultra-high reliability (e.g., $10^{-5}$ to $10^{-8}$ packet loss probability) and ultra-low end-to-end (E2E) delay (e.g, $1$~ms), URLLC {lays} the foundation for several mission-critical applications, such as industrial automation, Tactile Internet, remote driving, virtual reality (VR), and tele-surgery \cite{kang2018augmenting,zhao2018toward,Gerhard2014The}. How to achieve two conflicting requirements on delay and reliability remains an open problem.
%
%Different from traditional video or audio services in the , URLLC has much more stringent requirements on end-to-end (E2E) delay (e.g, $1$~ms) and reliability .
%
%Due to the stringent ultra-high reliability  and ultra-low . This is distinct from the traditional LTE system, where the delay of the air interface is around $10$~ms and its delay and reliability are not guaranteed by providing best effort services. There are great demands of URLLC services from many applications, such as . The main problem of URLLC is how to guarantee the ultra-low delay and ultra-high reliability requirements.

{To improve reliability, several technologies have been proposed in the existing literature and specifications, such as K-repetition \cite{3GPP2017Agree}, frequency hopping \cite{she2017radio}, large-scale antenna systems \cite{vu2017ultra}, and multi-connectivity \cite{nielsen2017ultra}}. With these technologies, different kinds of diversities are exploited to improve reliability at the cost of more radio resources. On the other hand, to reduce latency in the air interface, the short frame structure was proposed in 5G New Radio (NR)\cite{lin20185g}, {and fast uplink grant schemes were proposed to reduce access delay \cite{schulz2017latency,jacobsen2017system}}. However, there are some other delay components in the networks, such as delays in buffers of devices, computing systems, backhauls, and core networks. As a result, the user experienced delay can hardly meet the requirements of URLLC. Novel concepts and technologies that can reduce the user experienced delay and improve overall reliability (i.e., total packet losses and errors in different parts of the system) are in urgent need.

To tackle these challenges, we aim to meet the requirements of URLLC by jointly optimizing prediction and communication. The basic idea is to predict the future system states {at the transmitter}, such as locations and force feedback, and then send them to the receiver in advance. In this way, the user experienced delay can be reduced significantly. {For example, if the E2E delay is $10$~ms and the prediction horizon is $9$~ms, then the user experienced delay is $1$~ms.} However, predictions are not error-free, and long prediction horizon will lead to a large prediction error probability. Intuitively, there is a trade-off between the user experienced delay and the overall reliability. To satisfy the two conflicting requirements of URLLC, we need to jointly optimize the prediction and communication systems. Specifically, in this paper, we will address the following questions: \emph{1) how to characterize the tradeoff between user-experienced delay and overall reliability with prediction and communication co-design? 2) Is it possible to satisfy the requirements of URLLC by prediction and communication co-design? 3) If yes, how to maximize the number of URLLC services that can be supported by the system?}

%The most challenging part is that the prediction is not error-free

% As a fact, when the prediction horizon become longer, The prediction tends to be more unreliable. Although predicting the future system state in a longer time may relax the E2E communication delay requirements, the larger prediction errors due to longer prediction horizon may cancel the benefits from relaxed E2E communication delay requirements. On the contrary, when the prediction horizon is small, the prediction could be very reliable. But the requested E2E communication delay is also small. The reliability of the communication system could be reduced due to the strict E2E communication delay requirement. As a result, there exists an optimal prediction horizon and E2E communication delay design. As such, it is necessary to jointly design the prediction and communication systems to reach an ultra-low delay with the ultra-high reliability requirement guaranteed.

The above questions are challenging to answer since multiple components of delay and errors are involved in prediction and communication systems. As such, we need a prediction and communication co-design framework which takes different delay components and errors into account. Moreover, the complicated constraints on the user experienced delay and the overall reliability are non-convex in general, and hence it is very difficult to find the optimal solution.

%\red{The above questions are challenging to answer because complicated heterogenous components of delay and reliability are involved in the considered prediction and communication co-design system, resulting from queueing, short packet transmissions, and predictions. Specifically, the delay components include the queueing delay, transmission delay, which could be traded off with the prediction horizon. The packet loss factors include the delay bound violation in queueing, decoding errors, and prediction errors. As such, we need a prediction and communication co-design framework which takes into consideration above delay components and packet loss factors, and characterizes their closed-form relationship between each delay (or time) component and packet loss factor during queueing, transmission and predictions, especially the relationship between prediction horizon and prediction error and their impacts on the latency and reliability of the co-design system.} {(I suggest removing the discussion on the delay components and packet loss probabilities here since they depend on the system model. The reviewers may say, why the delay components only include queueing delay and transmission delay? Why there are only three factors lead to packet loss? Do prediction errors also lead to packet loss?)}

\subsection{Our Contributions}
%To address the above issues, we jointly optimize prediction and communication systems for URLLC.
The main contributions of this paper are summarized as follows:

%\red{Too many details in introduction(where the reliability of predictions decreases with the prediction horizon, and the reliability of communications increases with the queueing delay or transmission delay)}. We establish a \red{quantitative prediction model (only in an example?)}, which characterizes the close-form relationship between prediction error and prediction horizon. \red{We also analyze the relationship between the reliability of the communication system and the E2E communication delay, i.e., the closed-form relations between queueing delay and queueing delay violation probability, and between transmission delay and decoding error probability. (These results are obtained in previous works)} Based on above models and analyses, we minimize the required bandwidth to guarantee the users experienced delay and reliability requirements by optimizing the queueing delay, transmission delay and prediction horizon. \red{Although there are some research used prediction in URLLC \cite{hou2018Burstiness,strodthoff2018enhanced,li2018predictive,antonakoglou2018towards,sakr2009prediction,sakr2011human,simsek20165g,tong2018minimizing}, they did not quantitatively co-design of prediction and communications taking account of the prediction errors. To the best knowledge of the authors, this is the first work to quantitatively co-design of prediction and communications for URLLC. (Move to related works?)}

\begin{itemize}
\item We establish a framework for prediction and communication co-design, where the time and frequency resource allocation in the communication system and the prediction horizon in the prediction system are jointly optimized to maximize the number of devices that can be supported in the system.

\item We derive the closed-form expressions of the decoding error probability, the queueing delay violation probability, prediction error probability, and analyzed their properties. From these results, the tradeoff between user experienced delay and overall reliability can be obtained.

\item We propose an algorithm to find a near optimal solution of the optimization problem. The performance loss of the near optimal solution is studied and further validated via numerical results. Besides, we analyze the complexity of the algorithm, which linearly increases with the number of devices.
\end{itemize}

Furthermore, to evaluate the performance of the proposed method, we compared it with a benchmark solution without prediction. Simulation results show that the tradeoff can be improved remarkably with prediction and communication co-design. In addition, an experiment is carried out to validate the accuracy of mobility prediction in practical remote-control scenarios.

The rest of this paper is organized as follows: In  Section II, we review the related literature. The system model is presented in Section III. The co-design of prediction and communication is proposed in Section IV. Numerical and experimental results are presented in Section V, and conclusions are drawn in Section VI.

\section{Related Work}
\subsection{Communications in URLLC}
There are some existing solutions to reduce latency in communication systems for URLLC \cite{lin20185g,soret2014fundamental,schulz2017latency,jacobsen2017system,sachs20185g}. With the 5G New Radio (NR) \cite{lin20185g}, the notion of ``mini-slot" is introduced to support transmissions with the delay as low as the duration of a few symbols. The queueing delay is analyzed and optimized in \cite{soret2014fundamental}, where the tradeoff among throughput, delay and reliability was studied. {To reduce the access delay in uplink transmissions, a semi-persistent scheduling (SPS) scheme was developed in \cite{schulz2017latency}.} A grant-free protocol was proposed in \cite{jacobsen2017system} to further avoid the delay caused by scheduling requests and transmission grants. {With the preemptive scheduling scheme in \cite{sachs20185g}, the short packets with high priority can preempt an ongoing long packet transmission without waiting for the next scheduling period. With this scheme, the scheduling delay of short packets is reduced.}

To improve the reliability for the low latency communications, different kinds of diversities were introduced \cite{3GPP2017Agree,she2017radio,vu2017ultra,nielsen2017ultra}. In \cite{3GPP2017Agree}, K-repetition was proposed to avoid retransmission feedback. The basic idea is to send multiple copies of each packet without waiting for the acknowledgment feedback. Considering that the required delay is shorter than channel coherence time, frequency hopping was adopted in \cite{she2017radio} to improve reliability. In \cite{vu2017ultra}, a Lyapunov optimization problem was formulated to improve the reliability with guaranteed latency, where spatial diversity was used to improve reliability. In \cite{nielsen2017ultra}, interface diversity was proposed to achieve URLLC without modifications in the baseband designs by providing multiple communication interfaces. However, by introducing diversities, the reliability is improved at the cost of low resource utilization efficiency.

This tradeoff between delay and reliability has been exhaustively studied in communication systems \cite{yang2014quasi,bennis2018ultra,she2018cross,berardinelli2018reliability}. To reduce the transmission delay, the blocklength of channel codes is short, and the decoding error probability is nonzero for arbitrary signal-to-noise ratio (SNR). The fundamental tradeoff between transmission delay and decoding error probability in the short blocklength regime was derived in \cite{yang2014quasi}. The tradeoff between the queueing delay and the delay bound violation probability was studied in \cite{bennis2018ultra}. To achieve a lower delay bound, the violation probability increases. Moreover, grant-free schemes can help reduce latency, but introduce extra packet losses due to transmission collisions. How to achieve ultra-high reliability with grant-free schemes was studied in \cite{berardinelli2018reliability}and it is shown that the proposed stop-and-wait protocol can achieve $10^{-5}$ outage probability.

%Moreover, even if diversities and thus more resources are used, the URLLC requirements still cannot be met by only designing and optimizing the air interface for many scenarios because of the unavoidable backhaul, routing and propagations delays.

\subsection{Predictions in URLLC}
To achieve satisfactory delay and reliability in URLLC, different kinds of predictions have been studied in the existing literature \cite{tong2018minimizing,simsek20165g,hou2018Burstiness,li2018predictive,makki2018fast,strodthoff2018enhanced}.

In \cite{tong2018minimizing}, the predicted control commands were sent to the receiver and waiting in the buffer. When a control command is lost in communications, predicted commands in the receiver's buffer will be executed. The length of predictive control commands was optimized to minimize the resource consumption. The idea of model-mediated tele-operation approach was mentioned in \cite{simsek20165g}. By predicting the movement or the force feedback, the user experienced delay can be reduced. In both \cite{tong2018minimizing} and \cite{simsek20165g}, prediction errors were not considered, and whether we can achieve ultra-high reliability in the systems remains unclear.

Different from command or mobility predictions in control systems, predicting some other features of traffic or performance of communications is also helpful. In \cite{hou2018Burstiness}, based on the predicted traffic state, a bandwidth reservation scheme was proposed to improve the spectral efficiency of URLLC. By exploiting the correlation among different nodes, the behavior of different users can be predicted \cite{li2018predictive}. Then, by reserving resources according to the predicted behavior, the access delay can be reduced. {A fast hybrid automatic repeat request (HARQ) protocol was proposed in \cite{makki2018fast}, prediction is used to omit some HARQ feedback signals and successive message decodings, so that the expected delay can be improved by $27\%$ to $60\%$ compared with standard HARQ.} In \cite{strodthoff2018enhanced}, the outcome of the decoding was predicted before the end of the transmission. With the predicted result, there is no need to wait for the acknowledgment feedback, and thus the E2E delay can be reduced.

\section{System Model}
\begin{figure}[htbp]
	%\vspace{-0.1cm}
	\centering
	\begin{minipage}[t]{0.95\textwidth}
		\includegraphics[width=1\textwidth]{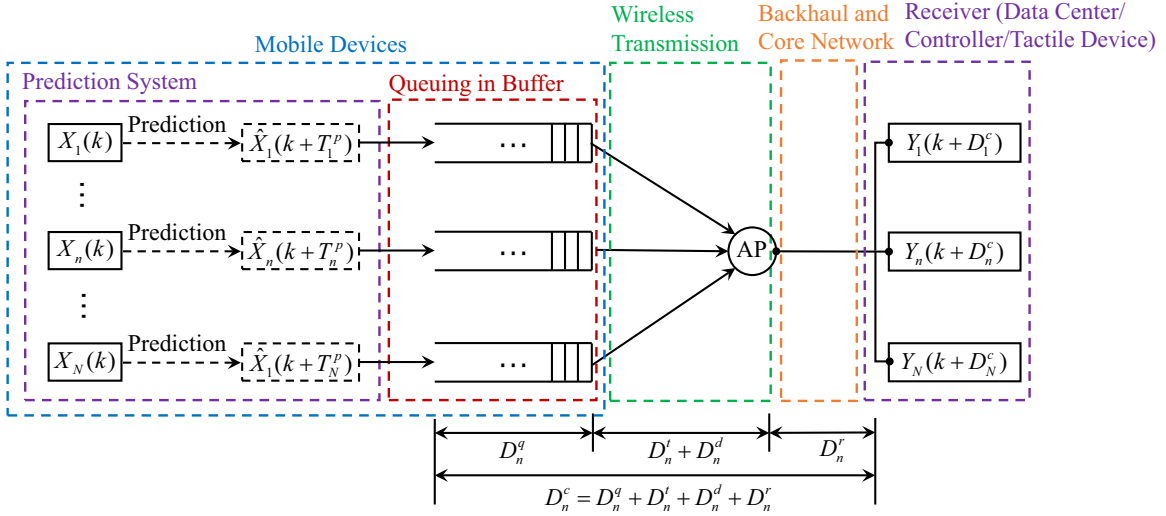}
	\end{minipage}
	%\vspace{-0.1cm}
	\caption{{Illustration of network structure.}}
	\label{fig:NetStc}
	\vspace{-0.3cm}
\end{figure}

%{The proposed framework is designed to be applied in multiple application scenarios, such as remote driving, industrial automation, Tactile Internet and so on. In remote driving \cite{kang2018augmenting}, a human driver can remotely control a vehicle based on the feedback from various sensors installed on the vehicle. Remote driving can be used independently, or as a backup on failure of self-driving functionality. In industrial automation \cite{zhao2018toward}, sensors update information to the controller to perform better closed-loop control,
%or to a data center for monitoring or fault detection. In Industry 4.0, the data center or controller can be
%deployed at the edge of the network or in a cloud. In Tactile Internet \cite{Gerhard2014The}, force and torques are sent to a tactile device to render the sense of touch, and thus can enable haptic communications.}

As shown in Fig. \ref{fig:NetStc}, we consider a joint prediction and communication system, where $N$ mobile devices send packets to {a receiver, which could be data center, controller, or tactile device}. {The function of the receiver depends on specific applications. In remote driving \cite{kang2018augmenting}, a human driver can remotely control a vehicle based on the feedback from various sensors installed on the vehicle. In factory automation \cite{zhao2018toward}, sensors update information to the controller to perform better closed-loop control, or to a data center for monitoring or fault detection. In Tactile Internet \cite{Gerhard2014The}, force and torques are sent to a tactile device to render the sense of touch, and thus can enable haptic communications.} The packets generated by each device may include different features, such as the location, velocity and acceleration of a device in remote driving or industrial automation, or the force and torques in Tactile Internet.

{The receiver can be deployed at a mobile edge computing (MEC) server or a cloud center. In our framework, we consider a general wireless communication system, where mobile devices send packets to a cloud center via wireless links, backhauls, and core networks. The framework is also suitable for an MEC system, where the delays and packet losses in backhauls and core networks are set to be zero \cite{Changyang2019IoT}.}

%\red{Whether the packets are periodic or event-driven has no impact on prediction and communication co-design. So, I removed this paragraph here.}

%The number of packets generated in each slot depends on the mobility of the device and the random events detected by the device.
%
%
%
%The packets produced in the transmitted can be random in the event-driven case or deterministic in the periodic case. In the event-driven case, the packets arrive randomly, and thus the packet arrival interval may be smaller than the E2E delay. As a result, there may be a queue in the buffer of each device, which is the scenario considered in this paper. In the periodic event case, the packets arrive in a deterministic and periodic way. In that case, the system can be designed without queueing, which can be regarded as a special case and has been studied in our conference version \cite{hou2018ultra}. In this journal version, we mainly focus on the event-driven case with multiple users, which is more challenge because the unavoidable queueing delay and queueing delay violation probability should be considered in the URLLC design.

\subsection{User Experienced Delay}
Time is discretized into slots. The duration of each slot is denoted as $T_s$. Let $X_n(k)=[x^1_n(k),x^2_n(k),...,x^F_n(k)]^T$ be the state of the $n$th device in the $k$th slot, where $F$ is the number of features. The state of the $n$th device that is received by the receiver in the $k$th slot is denoted as $Y_n(k)$. In traditional communication systems, each device sends its current state $X_n(k)$ to the data center. Let $D^{\rm c}_n$ (slots) be {the $n$th device's} end-to-end (E2E) delay in the communication system. If the packet that conveys $X_n(k)$ is decoded successfully in the $(k+D^{\rm c}_n)$th slot, then $Y_n(k+D^{\rm c}_n)=X_n(k)$, and the user experienced delay is $D^{\rm c}_n$. {For clarification, the key notations are listed in Table \ref{table:key_notations}.}

\begin{table}[!t]
\renewcommand{\arraystretch}{1}
\caption{{Index of Key Notations}}
\label{table:key_notations}
\centering
{\begin{tabular}{|c|l|}
\hline
\textbf{Notation}                                      & \textbf{Description}        \\
\hline
$N$                 & number of mobile devices         \\
\hline
$F$                 & number of features in a state         \\
\hline
$K_n$               & number of copies transmitted in  $K$-Repetition        \\
\hline
$T_s$               & duration of each time slot        \\
\hline
$D^{\rm c}_n$       & end-to-end(E2E) delay in communication system         \\
\hline
$D^{\rm c}_n$       & end-to-end(E2E) delay in communication system         \\
\hline
$T^p_n$             & prediction horizon of the $n$th device     \\
\hline
$D^{\rm e}_n$       & delay experienced by the $n$th device  \\
\hline
$D^{\rm q}_n$       & queueing delay of the $n$th device    \\
\hline
$D^{\rm t}_n$       & transmission delay of the $n$th device   \\
\hline
$D^{\rm d}_n$       & decoding delay of the $n$th device \\
\hline
$D^{\rm r}_n$       &  delay in backhauls and core networks of the $n$th device   \\
\hline
$D^{\tau}_n$        & transmission duration of each copy in $K$-Repetition of the $n$th device\\
\hline
$D_{\rm max}$        & delay requirement\\
\hline
$\varepsilon^{\rm p}_n$      &  prediction error probability of the $n$th device   \\
\hline
$\varepsilon^{\rm q}_n$      &  queueing delay bound violation probability of the $n$th device   \\
\hline
$\varepsilon^{\rm t}_n$      &  packet loss probability of the $n$th device   \\
\hline
$\varepsilon^{\tau}_n$      &  decoding error probability of the $n$th device   \\
\hline
$\bar{\varepsilon}^{\tau}_n$      &  expected decoding error probability of the $n$th device   \\
\hline
$\varepsilon^{\rm o}_n$      &   overall reliability of the $n$th device   \\
\hline
$\varepsilon_{\rm max}$      &   reliability requirement   \\
\hline
$X_n(k)$            & state of the $n$th device in the $k$th slot          \\
\hline
$\hat{X}_n(k)$      & predicted state of the $n$th device in the $k$th slot\\
\hline
$Y_n(k)$            & received state of the $n$th device in the $k$th slot       \\
\hline
$W_n(k)$            & transition noise of the $n$th device  in the $k$th slot       \\
\hline
$E_n(k)$            & difference between real state and predicted state of the $n$th device  in the $k$th slot      \\
\hline
$\Phi_n$            & state transition matrix of the $n$th device       \\
\hline
$E^{\rm B}_n$       & effective bandwidth of the $n$th device \\
\hline
$\lambda_n$         & average packet arrival rate of the $n$th device \\
\hline
$B_n$               & bandwidth of the $n$th device \\
\hline
$\eta$              & fraction of time and frequency resources for
data transmission\\
\hline
$P^{\rm t}_n$       & transmit power of the $n$th device \\
\hline
$N_0$               & noise power spectral density \\
\hline
$\gamma_n$          & SNR of the $n$th device \\
\hline
$a_n$               & large-scale channel gain of the $n$th device \\
\hline
$g_n$               & small-scale channel gain of the $n$th device \\
\hline
$\vartheta$         & SNR loss due to inaccurate channel estimation \\
\hline
$f^{-1}_{\rm Q}(\cdot)$         & inverse function of the Q-function \\
\hline
$N_{\rm r}$         &  number of antennas at the AP \\
\hline
\end{tabular}
}
\end{table}

where $B_n$ is the bandwidth, $P^{\rm t}_n$ represents the transmit power, $N_0$ denotes the noise power spectral density, $\gamma_n=\frac{a_n g_n P^{\rm t}_n}{\vartheta N_0 B_n}$ represents the received SNR, $a_n$ denotes the large-scale channel gain, $g_n$ is the small-scale channel gain, $\vartheta>1$ is the SNR loss due to inaccurate channel estimation, $V_n= 1 - [ 1+ \gamma_n ]^{-2}$ \cite{yang2014quasi}, $f^{-1}_{\rm Q}(\cdot)$ is the inverse function of the Q-function, and $\varepsilon^{\tau}_n$ is the decoding error probability. The blocklength of channel codes is $\eta D^{\tau}_n T_s B_n$. When the blocklength is large, \eqref{eq:comm_delay} approaches the Shannon capacity.

\begin{figure}[htbp]
	%\vspace{-0.1cm}
	\centering
	\begin{minipage}[t]{0.7\textwidth}
		\includegraphics[width=1\textwidth]{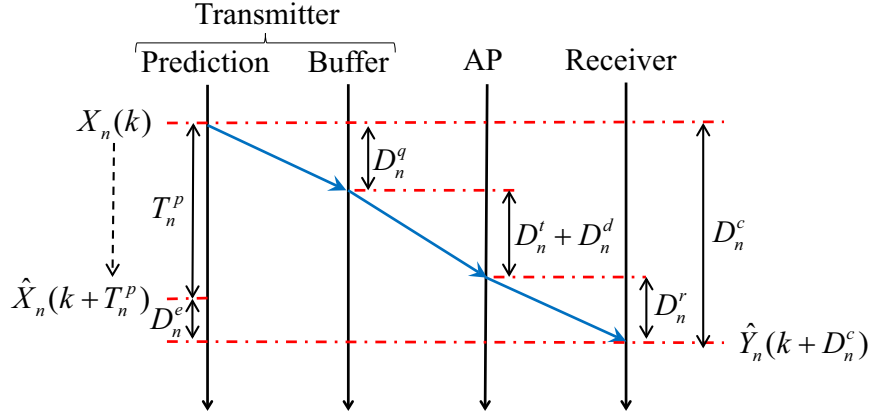}
	\end{minipage}
	%\vspace{-0.1cm}
	\caption{{Illustration of prediction and communication co-design.}}
	\label{fig:PredictStc}
	%\vspace{-0.1cm}
\end{figure}

As shown in Fig. \ref{fig:PredictStc}, to improve the user experienced delay, each device predicts its future state. $T^p_n$ is denoted as the prediction horizon. In the $k$th slot, the device generates a packet based on the predicted state $\hat{X}_n(k+T^p_n)$. After $D^{\rm c}_n$ slots, the packet is received by the data center. Then, we have $Y_n(k+D^{\rm c}_n)=\hat{X}_n(k+T^p_n)$, which is equivalent to $Y_n(k)=\hat{X}_n\left[k-(D^{\rm c}_n-T^p_n)\right], \forall k$. Therefore, the delay experienced by the user is $D^{\rm e}_n=D^{\rm c}_n-T^p_n$.\footnote{If $D^{\rm c}_n$ is smaller than $T^p_n$, $D^{\rm e}_n$ is negative. This means that the receiver can predict the states of devices. In this paper, we only consider the scenario that $D^{\rm e}_n \ge 0$.}

%The E2E  $D^{\rm c}_n$ is composed of  Since the E2E communication delay is $D^{\rm c}_n$, $\hat{X}_n(k+T^p_n)$ will arrive at the receiver in the $(k+D^{\rm c}_n)$th slot. If the packet is decoded successfully, then , i.e., $Y_n(k)=\hat{X}_n\left[k-(D^{\rm c}_n-T^p_n)\right]$ for arbitrary time $k$. Therefore, the delay experienced by the $n$th receiver is  as shown in Fig. \ref{fig:PredictStc}.

\begin{remark}
\emph{It is worth noting that the states of adjacent slots could be correlated. Thus, source coding schemes that compress the information in multiple slots can achieve higher compression ratio. On the other hand, channel coding schemes that encode the packets to be transmitted in multiple slots into one block, can achieve higher reliability. However, both of them will lead to a longer decoding delay. To achieve ultra-low latency, in this paper we assume that the source coding and channel coding in the $k$th slots only depend on $\hat{X}_n(k+T^p_n)$ and the data to be transmitted in this slot.}
\end{remark}

\subsection{Delay and Reliability Requirements}
The delay and reliability requirements are characterized by a maximum delay bound and a maximum tolerable error probability, $D_{\max}$ and $\varepsilon_{\max}$. It means that ${X}_n(k)$ should be received by the data center before the $(k+D_{\max})$th slot with probability $1-\varepsilon_{\max}$.

To satisfy the delay requirement, the user experienced delay should not exceed a maximal delay bound, i.e.,
\begin{equation}\label{eq:DelayReq}
  D^{\rm e}_n=D^{\rm c}_n-T^p_n \le D_{\rm max}.
\end{equation}
In the considered communication system, the E2E communication delay $D^{\rm c}_n$ includes queueing delay $D^{\rm q}_n$, transmission delay $D^{\rm t}_n$, decoding delay $D^{\rm d}_n$, and delay in backhauls and core networks $D^{\rm r}_n$.

Thus, the constraint in \eqref{eq:DelayReq} can be re-expressed as follows,
\begin{equation}\label{eq:comm_delay}
D^{\rm q}_n+D^{\rm t}_n+{D^{\rm d}_n}+D^{\rm r}_n-T^p_n \le D_{\rm max},
\end{equation}
{where $D^{\rm d}_n=\kappa D^{\rm t}_n$, $\kappa>0$.}

The overall reliability depends on prediction errors and packet losses in communications. In the control system, if the difference between the actual state of the device and the received state does not exceed a required threshold, the user cannot notice the difference. For example, in Tactile Internet, the minimum difference of the force stimulus intensity that our hands can percept is referred to as just noticeable difference (JND) \cite{feyzabadi2013human}. We define the difference between $\hat{X}_n(k)$ and $X_n(k)$ as $E_n(k) = [e^1_n(k),e^2_n(k),...,e^F_n(k)]^T$, where $e^j_n(k) = \hat{x}^j_n(k+D^{\rm e}_n)-x^j_n(k)$. The JND of this system is denoted as $\Delta=[\delta_{1},\delta_{2},...,\delta_{N}]^T$. Then, the prediction error probability is given by
\begin{equation}\label{eq:DefPredictErr2}
  \varepsilon^{\rm p}_n=1-\prod\limits_{j=1}^{N} \Pr\{|e^j_n(k)|\le \delta_{j}\},
\end{equation}

{Even if} $\hat{X}_n(k)$ is accurate enough, it will be useless if it is not received by the data center before the $(k+D_{\max})$th slot. Denote the queueing delay bound violation probability and the packet loss probability of the $n$th device as $\varepsilon^{\rm q}_n$ and $\varepsilon^{\rm t}_n$, respectively. Then, the overall reliability of the device can be expressed as follows,
\begin{equation}\label{eq:OverallReliability}
  \varepsilon^{\rm o}_n=1-(1-\varepsilon^{\rm q}_n)(1-\varepsilon^{\rm t}_n)(1-\varepsilon^{\rm p}_n).
\end{equation}
To achieve ultra-high reliability, all of $\varepsilon^{\rm q}_n$, $\varepsilon^{\rm t}_n$ and $\varepsilon^{\rm p}_n$ should be small (i.e., less than $10^{-5}$). Thus, \eqref{eq:OverallReliability} can be accurately approximated by $\varepsilon^{\rm o}_n \approx \varepsilon^{\rm q}_n+\varepsilon^{\rm t}_n + \varepsilon^{\rm p}_n$, and the reliability requirement can be satisfied if
\begin{equation}\label{eq:ReliableReq}
\varepsilon^{\rm q}_n + \varepsilon^{\rm t}_n + \varepsilon^{\rm p}_n \le \varepsilon_{\rm max}.
\end{equation}

\section{{Tradeoffs in Prediction and Communication Systems}} \label{section3}
In this section, we first consider a general linear prediction framework, and derive the relation between the prediction error probability and the prediction horizon in a closed form. Then, we characterize the tradeoff between communication reliability and E2E delay for short packet transmissions in a closed form. Based on the analysis, we further study how to maximize the number of URLLC services that can be supported by the system.

\subsection{State Transition Function}
We assume that the state of the $n$th device, $X_n(k)$, changes according to the following state transition function \cite{kay1993fundamentals}
\begin{equation}\label{eq:state_transition}
  X_n(k+1)=\Phi_n X_n(k) +W_n(k),
\end{equation}
where $\Phi_n=[\phi^{i,j}_n]_{F \times F}$, $i,j=1,2,\cdots,F$, is the state transition matrix and $W_n(k)=[w^{i}_n(k)]_{F \times 1}$, $i=1,2,\cdots,F$, is the transition noise. We assume that $\Phi_n$ is constant, and thus it can be obtained from measurements or physical laws. The elements of $W_n(k)$ are independent random variables that follow Gaussian distributions with zero mean and variances $\sigma_{1}^2,\sigma_{2}^2,\cdots,\sigma_{F}^2$, respectively.

%{{\bf{Remark:}}
\begin{remark}
{\emph{This model is widely adopted in kinematics systems or control systems \cite{kay1993fundamentals,klanvcar2007tracking}. Here we consider a general prediction method for a linear system. This is because for non-linear system, the relation between the prediction horizon and the prediction error probability can hardly be derived in a closed-form expression. To implement our framework in non-linear systems, data-driven prediction methods such as neural networks should be applied. These methods do not rely on system models, and will be considered in our future work.}}
\end{remark}

According to \eqref{eq:state_transition}, the state in the $(k+T^p_n)$th slot is given by
\begin{equation}\label{eq:n_step_real_state}
  X_n{(k+T^p_n)}=(\Phi_n)^{T^p_n} X_n(k) + \sum\limits_{i=1}^{T^p_n} (\Phi_n)^{T^p_n-i}W_n{(k+i-1)}.
\end{equation}

%of the process from the current state at $k$th slot to the future state at the $(k+1)$th slot, which is assumed to be constant over time and known from historical data.  is the associated unknown disturbances named state transition noises, each element of which is

\subsection{Prediction Horizon and Prediction Error Probability}
%{Before predicting future states, the device needs to estimate the current state. Due to the inaccuracy of the measurement, the device does not know the real state $X(k)$. The estimated state in the current slot (e.g., the $k$th slot) can be expressed as
%\begin{align}
%\tilde{X}(k)=X(k)+W(0), \label{eq:initialE}
%\end{align}
%where $W(0)=[w_{1}(0),w_{2}(0),...,w_{N}(0)]^T$ is the initial error vector. The elements of $W(0)$ are assumed to be Gaussian distributed independent random variables with zero mean and variances $\sigma_{1}^2(0),\sigma_{2}^2(0),\cdots,\sigma_{N}^2(0)$.} (\red{which are assumed to be much less than the state transition errors.}) (Comment: Do we need this assumption here? What are the state transition errors? Do you mean the variance of the transition noise?)

Inspired by Kalman filter, we consider a general linear prediction method \cite{kay1993fundamentals}. Based on the system state in the $k$th slot, we can predict the state in the $(k+1)$th slot according to following expression,
\begin{equation}\label{eq:prediction_model}
  \hat{X}_n{(k+1)}=\Phi_n X_n(k).
\end{equation}
From \eqref{eq:prediction_model}, we can further predict the state in the $(k+T^p_n)$th slot,
\begin{equation}\label{eq:n_step_prediction_model}
  \hat{X}_n{(k+T^p_n)}=(\Phi_n)^{T^p_n} X_n(k).
\end{equation}
After $T^p_n$ steps of prediction, the difference between $X_n{(k+T^p_n)}$ and $\hat{X}_n{(k+T^p_n)}$ can be derived as follows,
\begin{equation}\label{eq:n_step_prediction_error}
\begin{aligned}
E_n(k+T^p_n) &\triangleq X_n{(k+T^p_n)}-\hat{X}_n{(k+T^p_n)} \\
&=W_n{(k+T^p_n-1)} + \sum\limits_{i=1}^{T^p_n-1} (\Phi_n)^{T^p_n-i}W_n{(k+i-1)}.
\end{aligned}
\end{equation}

The $j$th element of $E_n(k+T^p_n)$ is given by
\begin{equation}\label{eq:e_j_k}
  e^j_n(k+T^p_n)=w^j_n{(k+T^p_n-1)}+\sum\limits_{i=1}^{T^p_n-1}\sum\limits_{m=1}^{F}\phi_{n,j,m,T^p_n-i}w^m_n{(k+i-1)},
\end{equation}
where $\phi_{n,j,m,T^p_n-i}$ is the element of $(\Phi_n)^{T^p_n-i}$ at the $j$th row and the $m$th column.

Since the state transition noises follow independent Gaussian distributions, and $e^j_n(k+T^p_n)$ is a linear combination of them, $e^j_n(k+T^p_n)$ follows a Gaussian distribution with zero mean. The variance of $e^j_n(k+T^p_n)$ is denoted as $\rho_{n,j}^2(T^p_n)$, which is given by
\begin{equation}\label{eq:sigma_n_j}
  \rho_{n,j}^2(T^p_n)=\sigma_{j}^2+\sum\limits_{i=1}^{T^p_n-1}\sum\limits_{m=1}^{F}(\phi_{n,j,m,T^p_n-i})^2\sigma_{m}^2.
\end{equation}

Therefore, $\Pr\{|e^j_n(k+T^p_n)|\le \delta_j\}$ can be derived as follows,
\begin{equation}\label{eq:DefPredictErr3}
\begin{aligned}
  \Pr\{|e^j_n(k+T^p_n)|\le \delta_j\} =& 1- \Pr\{|e^j_n(k+T^p_n)| > \delta_j\}\\
                     =& 1-\psi_{T^p_n,j}\left(-\delta_j\right)\\
                     =& 1-\psi\left(\frac{-\delta_j}{\rho_{n,j}(T^p_n)}\right),
\end{aligned}
\end{equation}
where $\psi_{T^p_n,j}(\cdot)$ is the cumulative distribution function (CDF) of $e^j_n(k+T^p_n)$, and $\psi(\cdot)$ is the CDF of standard Gaussian distribution with zero mean and unit variance.

By substituting \eqref{eq:DefPredictErr3} into \eqref{eq:DefPredictErr2}, $\varepsilon^{\rm p}_n$ can be expressed as follows,
\begin{equation}\label{eq:PredictErrCalc}
  \varepsilon^{\rm p}_n=1-\prod\limits_{j=1}^{F} \left[1-\psi\left(\frac{-\delta_j}{\rho_{n,j}(T^p_n)}\right)\right].
\end{equation}

From the expression in \eqref{eq:PredictErrCalc}, we can obtain the following property of $\varepsilon^{\rm p}_n$.

\begin{lemma}\label{prop:mono_epsilon_p}
$\varepsilon^{\rm p}_n$ strictly increases with the prediction horizon $T^p_n$.
\end{lemma}
\begin{proof}
Please see Appendix \ref{appendix:proposition:monop}.
\end{proof}

{Lemma \ref{prop:mono_epsilon_p} indicates that a longer prediction horizon leads to a larger prediction error probability. This is in accordance with the intuition. For example, predicting the mobility of a device in the next $100$~ms will be much harder than predicting the mobility in the next $10$~ms.}

%\subsection{Delay-Reliability Tradeoff for Short Packet Transmissions}
%\subsubsection{E2E Delay in the Communication System}
%As shown in Fig. \ref{fig:NetStc}, the distance between the device and the receiver can be large. In this case, The E2E delay includes queueing delays and transmission delays in radio access networks, routing delays in backhauls and core networks, propagation delays, and processing delays in multiple equipments.
%
%We are interested in optimization the delays in radio access networks. The queueing delay and the transmission delay are denoted as $D_{\rm q}$ and $D_{\rm t}$, respectively. The sum of all the other delays in the backhaul and the core network are assumed to be a constant $D_{\rm r}$. As such, the E2E delay of a communication system is
%\begin{equation}\label{eq:comm_delay}
%  D_{\rm c}=D_{\rm q}+D_{\rm t}+D_{\rm r}.
%\end{equation}

\subsection{Queueing Delay Bound Violation Probability}
To derive the queueing delay bound violation probability, $\varepsilon^{\rm q}_n$, we can use the concept of effective bandwidth \cite{she2018cross}. Effective bandwidth is defined as the minimal constant service rate of the queueing system that is required to ensure the maximum queueing delay bound and the delay bound violation probability \cite{chang1995effective}.\footnote{{To analyze the upper bound of the delay bound violation probability, a widely used tool is network calculus \cite{al2016network}. However, with network calculus, one can hardly obtain a closed-form expression of the delay bound violation probability. Since we are interested in the asymptotic scenarios that $\varepsilon^{\rm q}_n$ is very small, effective bandwidth can be used \cite{chang1995effective}.}}

The number of packets generated in each slot depends on the mobility of the device and the random events detected by the device. According to the observation in \cite{Condoluci2017Soft}, {packet} arrival processes in Tactile Internet are very bursty. To capture the burstiness of the packet arrival process, a switched Poisson process (SPP) can be applied \cite{hou2018Burstiness} \footnote{{In standardizations of 3GPP, In standardizations of 3GPP, queueing models are not specified since they depend on specific applications.}}. A SPP includes two traffic states. In each state, the packet arrival process follows a Poisson process. The average packet arrival rates are different in the two states, and the SPP switches between the two states according to a Markov chain. With the traffic state classification methods in \cite{hou2018Burstiness}, the AP knows the average packet arrival rate in the current state, $\lambda_n$ (packets/slot). According to \cite{she2018cross}, the effective bandwidth of the Poisson process is given by
\begin{equation}\label{eq:EB}
E^{\rm B}_n = \frac{ \ln{(1/\varepsilon^{\rm q}_n)}}{D^{\rm q}_n \ln \left[ \frac{ \ln{(1/\varepsilon^{\rm q}_n)}}{\lambda_n D^{\rm q}_n} + 1 \right]} \; \text{packets/slot},
\end{equation}
which is the minimal constant service rate required to ensure $D^{\rm q}_n$ and $\varepsilon_n^{\rm q}$. Since the transmission delay of each packet is fixed as $D^{\rm t}_n$, to guarantee the queueing delay violation probability, the following constraint should be satisfied,
\begin{equation}\label{eq:crosslayer}
  \frac{1}{D^{\rm t}_n} = E^{\rm B}_n.
\end{equation}

Then, the queueing delay violation probability can be derived as
\begin{equation}\label{eq:varepsilon_q}
  \varepsilon^{\rm q}_n = e^{ D^{\rm q}_n \phi(\lambda_n,E^{\rm B}_n)  },
\end{equation}
where
\begin{equation}\label{eq:varepsilon_q2}
  \phi(\lambda_n,E^{\rm B}_n)= E^{\rm B}_n \mathbb{W}_{-1}\left(-\frac{\lambda_n}{E^{\rm B}_n} e^{-\frac{\lambda_n}{E^{\rm B}_n}}\right) + \lambda_n,
\end{equation}
where $\mathbb{W}_{-1}(\cdot)$ is the ``$-1$" branch of the Lambert W-function, which is defined as the inverse function of $f(x)=x e^{x}$. The derivations of \eqref{eq:varepsilon_q} and \eqref{eq:varepsilon_q2} are given in Appendix \ref{appendix:lambert}.

With the expressions in \eqref{eq:varepsilon_q} and \eqref{eq:varepsilon_q2}, we can obtain the following property of $\varepsilon^{\rm q}_n$.
\begin{lemma}\label{prop:mono_epsilon_q}
$\varepsilon^{\rm q}_n$ strictly decreases with the queueing delay $D^{\rm q}_n$ when $\lambda_n$ and $E^{\rm B}_n$ are given.
\end{lemma}
\begin{proof}
Please see Appendix \ref{appendix:proposition:monoq}.
\end{proof}

{Lemma \ref{prop:mono_epsilon_q} indicates that with the same packet arrival process and service process, the queueing system with a smaller queueing delay bound requirement has a larger queueing delay violation probability. The intuition is that for a given CDF of the steady state queueing delay, the queueing delay violation probability decreases with the queueing delay bound.}

%To ensure the queueing delay requirements, the physical layer should support at least $E^{\rm B}_n$ (packets/slot). In physical layer, $K$-Repetition is used and thus each packet is transmitted for K copies, which a total transmission delay $D^{\rm t}_n=K D^{\tau}_n$. As such, it will take $D^{\rm t}_n$ slots for the physical layer to transmit one packet of the queue in the buffer. Or equally, in each slot, $1/(D^{\rm t}_n T_{\rm s})$ packets can be removed from the queue, i.e., a service rate of $1/(D^{\rm t}_n T_{\rm s})$ (pakcets/slot) from the perspective of the queueing. So we have

\subsection{Packet Loss Probability in Transmissions}
With predictions, the communication delay can be longer than the required delay bound $D_{\rm max}$ (e.g., $1$~ms). As such, retransmissions or repetitions becomes possible. To avoid feedback delay caused by retransmissions, we apply $K$-Repetitions to reduce the packet loss probability in the communication system, i.e., the device sends $K$ copies of each coding block no matter whether the first few copies are successfully decoded or not \cite{3GPP2017Agree}. The transmission duration of each copy is denoted as $D^{\tau}_n$. Then, we have $D^{\tau}_n = D^{\rm t}_n/K_n$. Some time and frequency resources are reserved for channel estimation at the AP. The fraction of time and frequency resources for data transmission is denoted as $\eta < 1$. To avoid overhead and extra delay caused by channel estimation at the device, we assume the device does not have channel state information {(CSI)}. {The impacts of CSI and training pilots on the achievable rate have been studied in the short blocklength regime \cite{schiessl2018delay,ostman2018short,mousaei2017optimizing,hassibi2003much}. If more resource blocks are occupied by pilots, the accuracy of the estimated CSI can be improved. However, the remaining resource blocks for data transmission reduces. How to allocate radio resources for pilots and data transmissions is a complicated problem and deserves further study. By assuming CSI is not available at the transmitters, our approach can serve as a benchmark for future research.}

For the transmission of each copy, we assume that the transmission duration is smaller than the channel coherence time and the bandwidth is smaller than the coherence bandwidth. This assumption is reasonable for short packet transmissions in URLLC. Then, the achievable rate in the short blocklength regime {over a quasi-static SIMO channel} can be accurately approximated by {the following normal approximation} \cite{yang2014quasi}\footnote{{The bounds of the decoding error probability can be obtained by using saddlepoint method \cite{lancho2019saddlepoint}, which is very accurate but has no closed-form expression. Since the gap between the normal approximation and practical coding schemes is around $0.1$~dB \cite{shirvanimoghaddam2018short}, it is accurate enough for our framework.}},
{\begin{equation}\label{eq:bound}
b_n \approx \frac{\eta D^{\tau}_n T_s B_n}{\ln{2}}\left[ \ln{\left( 1 + \gamma_n \right)} - \sqrt{\frac{V_n}{\eta D^{\tau}_n T_s B_n}} f^{-1}_{\rm Q}(\varepsilon^{\tau}_n) \right]\, \; (\rm bits/block),
\end{equation}}
where $B_n$ is the bandwidth, {$\gamma_n$ represents the received SNR},  $V_n= 1 - [ 1+ \gamma_n ]^{-2}$ \cite{yang2014quasi}, $f^{-1}_{\rm Q}(\cdot)$ is the inverse function of the Q-function, and $\varepsilon^{\tau}_n$ is the decoding error probability. The blocklength of channel codes is $\eta D^{\tau}_n T_s B_n$. When the blocklength is large, \eqref{eq:bound} approaches the Shannon capacity \footnote{{The results in \cite{schiessl2015delay} indicate that if Shannon capacity is used in the analyses, the delay bound and delay bound violation probability will be underestimated. Thus, the requirements of URLLC cannot be satisfied.}}.

According to \eqref{eq:bound}, the expected decoding error probability of each transmission over the SIMO channel is given by \cite{yang2014quasi}
\begin{equation}\label{eq:avg_epsilon_t}
\bar{\varepsilon}^{\tau}_n = \int_0^\infty f_{\rm Q}  \left\{ \sqrt{\frac{\eta D^{\tau}_n T_s B_n}{V_n}} \left[ \ln{\left( 1 + \frac{a_n g_n P^{\rm t}_n}{\vartheta N_0 B_n} \right)}- \frac{b_n \ln{2}}{\eta D^{\tau}_n T_s B_n} \right] \right\} \cdot f_g(x) dx,
\end{equation}
{where $\gamma_n = \frac{a_n g_n P^{\rm t}_n}{\vartheta N_0 B_n}$ is applied, $a_n$ denotes the large-scale channel gain, $g_n$ is the small-scale channel gain, $P^{\rm t}_n$ represents the transmit power, $\vartheta>1$ is the SNR loss due to inaccurate channel estimation, $N_0$ denotes the noise power spectral density, and $f_g(x)$ is the distribution of the instantaneous channel gain.} For Rayleigh fading channel, we have $f_g(x)=\frac{1}{(N_{\rm r}-1)!}x^{N_{\rm r}-1}e^{-x}$, where $N_{\rm r}$ is the number of antennas at the AP. From the approximation in \cite{She2018availability}\footnote{{As validated in \cite{She2018availability}, the approximation in \eqref{eq:avg_epsilon_t_approx} is accurate, especially when the number of antennas is large or the packet loss probability is small.}}, $\bar{\varepsilon}^{\tau}_n$ can be accurately approximated by
\begin{equation}\label{eq:avg_epsilon_t_approx}
  \bar{\varepsilon}^{\tau}_n = \frac{\omega_n a_n P^{\rm t}_n \sqrt{\eta D^{\tau}_n T_s B_n}}{\vartheta N_{0}B_n}\left[ (g^{\rm U}_n-g^{\rm L}_n) - \sum\limits_{i=0}^{N_{\rm r}}(N_{\rm r}-i) A^{\rm i}_n \right],
\end{equation}
where $\omega_n=\frac{1}{2\pi\sqrt{2^{2r^{\rm c}_n-1}}}$, $r^{\rm c}_n=\frac{b_n}{\eta D^{\tau}_n T_s B_n}$ is the number bits in each coding block, $g^{\rm U}_n=\frac{\vartheta N_{0}B_n\xi_n}{a_n P^{\rm t}_n}$, $g^{\rm L}_n=\frac{\vartheta N_{0}B_n\zeta_n}{a_n P^{\rm t}_n}$, $A^{\rm i}_n = \frac{(g^{\rm L}_n)^i }{i!}e^{-g^{\rm L}_n}-\frac{(g^{\rm U}_n)^i }{i!}e^{-g^{\rm U}_n}$, $\xi_n=\theta_n+\frac{1}{2\omega_n\sqrt{\eta D^{\tau}_n T_s B_n}}$, $\zeta_n=\theta_n-\frac{1}{2\omega_n\sqrt{\eta D^{\tau}_n T_s B_n}}$, and $\theta_n=2^{r^{\rm c}_n-1}$.

After $K$ repetitions, the packet loss probability in the communication system is given by
\begin{equation}\label{eq:epsilon_t}
\varepsilon^{\rm t}_n = (\bar{\varepsilon}^{\tau}_n)^{K_n}.
\end{equation}

From \eqref{eq:epsilon_t}, we can obtain the following property of $\varepsilon^{\rm t}_n$.
\begin{lemma}\label{prop:mono_epsilon_t}
When $D^{\tau}_n$ is given, $\varepsilon^{\rm t}_n$ strictly decreases with the repetition time $K_n$.
\end{lemma}
\begin{proof}
When $D^{\tau}_n$ is given, $\bar{\varepsilon}^{\tau}_n$ is fixed. According to \eqref{eq:epsilon_t}, $\varepsilon^{\rm t}_n$ decreases with $K_n$ since $\bar{\varepsilon}^{\tau}_n<1$.
\end{proof}

{Lemma \ref{prop:mono_epsilon_t} indicates that there is a tradeoff between the transmission delay and the reliability in communications. $K$-Repetition can be used to improve the transmission reliability at the cost of increasing the transmission delay.}

\section{{Prediction and Communication Co-design}} \label{section4}
{In the above tradeoff analyses, we obtained closed-form relations between each delay component (or prediction horizon) and its corresponding packet loss factor in terms of prediction, queueing and wireless transmission, respectively. Based on the above analyses, the tradeoff between the overall reliability and prediction horizon is revealed. As such, we could formulate the optimization problem in the following subsection.}

\subsection{Problem Formulation}
To maximize the number of devices that can be supported by the system, we optimize the delay components, prediction horizon, and bandwidth allocation of wireless networks. The optimization problem can be formulated as follows,

%\newpage

\begin{align}
\mathop {\mathop {\max }\limits_{{D^{\rm q}_n,D^{\rm t}_n,T^p_n,B_n,}} }\limits_{{n=1,...,N,}}
  \quad &N \label{eq:OptProblem1}\\
  \text{s.t.}
  \quad
    &\sum\limits_{n=1}^{N} B_n \le B_{\rm max}, \label{eq:OptB}\tag{\theequation a}\\
    &D^{\rm q}_n+D^{\rm t}_n+{D^{\rm d}_n}+D^{\rm r}_n-T^p_n \le D_{\rm max}, \label{eq:OptRelation_G}\tag{\theequation b}\\
    &\varepsilon^{\rm q}_n+\varepsilon^{\rm t}_n + \varepsilon^{\rm p}_n \le \varepsilon_{\rm max}, \label{eq:OptCommDelay_G}\tag{\theequation c}\\
    &\varepsilon^{\rm q}_n=\exp{\left\{ D^{\rm q}_n \left[\frac{\mathbb{W}_{-1}(-\lambda_n D^{\rm t}_n e^{-\lambda_n D^{\rm t}_n})}{D^{\rm t}_n} +\lambda_n\right] \right\}}, \label{eq:varepsilon_q_new}\tag{\theequation d}\\
    &\varepsilon^{\rm t}_n=\left\{\frac{\omega_n a_n P^{\rm t}_n \sqrt{\eta D^{\tau}_n T_s B_n}}{\vartheta N_{0}B_n}\left[ (g^{\rm U}_n-g^{\rm L}_n) - \sum\limits_{i=0}^{N_{\rm r}}(N_{\rm r}-i) A^{i}_n \right]\right\}^{K_n}, K_n D^{\tau}_n=D^{\rm t}_n \label{eq:varepsilon_t_new}\tag{\theequation e}\\
    &\varepsilon^{\rm p}_n=1-\prod\limits_{j=1}^{F} \left[1-\psi\left(\frac{-\delta_j}{\sqrt{\sigma_{j}^2+\sum\limits_{i=1}^{T^p_n-1}\sum\limits_{m=1}^{F}(\phi_{n,j,m,T^p_n-i})^2\sigma_{m}^2}}\right)\right], \label{eq:varepsilon_p_new}\tag{\theequation f}\\
    &n=1,2,3,\cdots,N, \label{eq:n_range}\tag{\theequation g}
\end{align}
where \eqref{eq:OptB} is the constraint on total bandwidth, \eqref{eq:OptRelation_G} is the constraint on user experienced delay, \eqref{eq:OptCommDelay_G} is the constraint on reliability. \eqref{eq:varepsilon_q_new} is obtained by substituting \eqref{eq:varepsilon_q2} and  \eqref{eq:crosslayer} into \eqref{eq:varepsilon_q}, \eqref{eq:varepsilon_t_new} is obtained from \eqref{eq:avg_epsilon_t_approx} and \eqref{eq:epsilon_t}, and \eqref{eq:varepsilon_p_new} is obtained by substituting \eqref{eq:sigma_n_j} into \eqref{eq:PredictErrCalc}.

Problem \eqref{eq:OptProblem1} is not a deterministic optimization problem since the numbers of optimization variables and constraints depend on the number of users, which is not given. In addition, some optimization variables are integers and the constraints in \eqref{eq:OptCommDelay_G}, \eqref{eq:varepsilon_q_new}, and \eqref{eq:varepsilon_t_new} are non-convex. Thus, it is very challenging to solve this problem.

\subsection{Algorithm for Solving Problem \eqref{eq:OptProblem1}}
To solve the problem \eqref{eq:OptProblem1}, we first find the minimal bandwidth $B_n$ required for each user to ensure its delay and reliability requirements, i.e., $(D_{\rm max},\varepsilon_{\rm max})$. By minimizing the bandwidth allocated to each user, the total number of users that can be supported with a given amount of total bandwidth can be maximized. Without the constraint on total bandwidth, the problem \eqref{eq:OptProblem1} can be decomposed into multiple single-user problems:
\begin{align}
\min\limits_{D^{\rm q}_n,D^{\rm t}_n,T^p_n}
  \quad &B_n \label{eq:OptProblem_G}\\
  \text{s.t.}
  \quad
    &\eqref{eq:OptRelation_G},\eqref{eq:OptCommDelay_G},\eqref{eq:varepsilon_q_new},\eqref{eq:varepsilon_t_new}\ \rm{and}\ \eqref{eq:varepsilon_p_new}.
\end{align}

To solve the above problem, we need the minimal bandwidth that is required to ensure a certain overall reliability. We denote it as $B_n^{\rm min}(\varepsilon_{n}^{\rm o})$. However, deriving the expression of  $B_n^{\rm min}(\varepsilon_{n}^{\rm o})$ is very difficult. To overcome this difficulty, we first minimize $\varepsilon_{n}^{\rm o}$ for a given $B_n$. Then, we find the minimal required bandwidth that can satisfy $\varepsilon_{n}^{\rm o} \leq \varepsilon_{\max}$ via binary search.

When $B_n$ is given, the minimal overall error probability can be obtained by optimizing $T^p_n$ in solving the following problem,
\begin{align}
\varepsilon_{n}^{\rm o,\rm min}(B_n)=\min\limits_{D^{\rm q}_n,D^{\rm t}_n,T^p_n}
  \quad &\varepsilon^{\rm q}_n+\varepsilon^{\rm t}_n+\varepsilon^{\rm p}_n \label{eq:OptProblem2_G}\\
  \text{s.t.}
  \quad
  &\eqref{eq:OptRelation_G},\eqref{eq:varepsilon_q_new},\eqref{eq:varepsilon_t_new}\ \rm{and}\ \eqref{eq:varepsilon_p_new}, \nonumber
\end{align}
For mathematical {tractability}, we set $\varepsilon^{\rm q}_n=\varepsilon^{\rm t}_n$. According to \cite{she2018cross}, this simplification leads to negligible performance loss. We will first prove $\varepsilon^{\rm q}_n$ and $\varepsilon^{\rm t}_n$ decreases with $T^p_n$ in the Proposition~\ref{prop:mono_n} when $\varepsilon^{\rm q}_n=\varepsilon^{\rm t}_n$.

\begin{proposition}\label{prop:mono_n}
$\varepsilon^{\rm q}_n$ and $\varepsilon^{\rm t}_n$ decrease with $T^p_n$ when $\varepsilon^{\rm q}_n=\varepsilon^{\rm t}_n$.
\end{proposition}
\begin{proof}
Please see Appendix \ref{appendix:proposition:mono_n}.
\end{proof}

{Proposition \ref{prop:mono_n} reveals the relation between the reliability of the queueing system (or the reliability of the wireless link) and the prediction horizon. With this relation, the number of independent optimization variables can be reduced.}

It can be recalled that $\varepsilon^{\rm p}_n$ increases with $T^p_n$. Thus, {together with Proposition 1}, the optimal solution is obtained when the equality in \eqref{eq:OptRelationEqual_G} holds, which is
\begin{equation}\label{eq:OptRelationEqual_G}
D^{\rm q}_n+D^{\rm t}_n+D^{\rm r}_n-T^p_n = D_{\rm max}.
\end{equation}
Moreover, for a given value of $T^p_n$, the values of $D^{\rm q}_n$ and $D^{\rm t}_n$ that satisfies $\varepsilon^{\rm q}_n=\varepsilon^{\rm t}_n$ and \eqref{eq:OptRelationEqual_G} can be obtained via binary search. Therefore, we only need to optimize $T^p_n$ in problem \eqref{eq:OptProblem2_G}. The optimal solution and the minimal overall reliability in this simplified scenario are denoted as $T_n^{p*}$ and $\varepsilon_{n}^{\rm o,\rm min *}(B_n)$, respectively.

%Unfortunately, the simplified problem is still non-convex. Since $D^{\rm q}_n+D^{\rm t}_n$ decreases with $T^p_n$ and $\varepsilon^{\rm q}_n+\varepsilon^{\rm t}_n$ decreases with $D^{\rm q}_n+D^{\rm t}_n$, $\varepsilon^{\rm q}_n+\varepsilon^{\rm t}_n$ decreases as $T^p_n$ increases. On the other hand, $\varepsilon^{\rm p}_n$ increases with $T^p_n$. A near optimal solution can be obtained when $\varepsilon^{\rm q}_n+\varepsilon^{\rm t}_n = \varepsilon^{\rm p}_n$. Further considering that $\varepsilon^{\rm q}_n=\varepsilon^{\rm t}_n$, we have $\varepsilon^{\rm q}_n=\varepsilon^{\rm t}_n = \varepsilon^{\rm p}_n/2$. The overall reliability achieved by this near optimal solution is denoted as $\hat{\varepsilon}_{n}^{\rm o,\rm min}(B_n)$. The gap between $\hat{\varepsilon}_{n}^{\rm o,\rm min}(B_n)$ and $\varepsilon_{n}^{\rm o,\rm min *}(B_n)$ is small and can be found in Proposition \ref{prop:error_bound}.

Unfortunately, the simplified problem is still non-convex. As such, we will propose an approximated solution as follows. According to {Lemma} \ref{prop:mono_epsilon_p}, $\varepsilon^{\rm p}_n$ increases with $T^p_n$, and we have proved $\varepsilon^{\rm q}_n$ and $\varepsilon^{\rm t}_n$ decreases with $T^p_n$ in Proposition \ref{prop:mono_n}. A near optimal solution can be obtained when $\varepsilon^{\rm q}_n+\varepsilon^{\rm t}_n = \varepsilon^{\rm p}_n$. Since the optimization variables are not integers, $\varepsilon^{\rm q}_n+\varepsilon^{\rm t}_n = \varepsilon^{\rm p}_n$ may not hold strictly. To address this issue, we can use binary search to find $\tilde{T}^p_n$ that satisfies $\varepsilon^{\rm p}_n \le 2\varepsilon^{\rm t}_n$ when $T^p_n \le \tilde{T}^p_n$, and $\varepsilon^{\rm p}_n> 2\varepsilon^{\rm t}_n$ when $T^p_n > \tilde{T}^p_n$. The corresponding reliability is denoted as $\hat{\varepsilon}_{n}^{\rm o,\rm min}(B_n)$. The overall reliability achieved by this near optimal solution is denoted as $\hat{\varepsilon}_{n}^{\rm o,\rm min}(B_n)$.

{The performance gap between the near optimal solution and optimal one is analyzed in the following Proposition \ref{prop:error_bound}.}
{\begin{proposition}\label{prop:error_bound}
The gap between $\hat{\varepsilon}_{n}^{\rm o,\rm min}(B_n)$ and $\varepsilon_{n}^{\rm o,\rm min *}(B_n)$ is less than $\varepsilon_{n}^{\rm o,\rm min *}(B_n)$, where $\varepsilon_{n}^{\rm o,\rm min *}(B_n)$ is the reliability achieved by the optimal solution.
\end{proposition}
\begin{proof}
Please see Appendix \ref{appendix:proposition:bound}.
\end{proof}
}

{Proposition \ref{prop:error_bound} shows that the gap between the near optimal overall reliability and the optimal one is bounded by the value of the optimal overall reliability. Since the optimal overall reliability is in the order of $10^{-5}$, the gap is very small.}

The required minimal bandwidth to guarantee the overall reliability can be obtained from the following optimization problem,
\begin{align}
\mathop {\min }\limits_{B_n} & \;B_n \label{eq:sub2}\\
\text{s.t.}\;& \hat \varepsilon _{n}^{\rm o,\rm min}(B_n) \leq \varepsilon_{\max}.\label{eq:reliable}\tag{\theequation a}
\end{align}
Since the packet loss in the communication system decreases with bandwidth, the optimal solution of problem \eqref{eq:sub2} is achieved when the equality in \eqref{eq:reliable} holds. Thus, the minimal bandwidth can be obtained via binary search. The algorithm to solve problem \eqref{eq:OptProblem_G} is summarized in Table \ref{T:algorithm}.

\subsection{Discussions on Implementation Complexity and Optimality}
%Denote the feasible region of bandwidth and prediction horizon as $[0,B_{\max}]$ and $[0,T^p_{\max}]$. Then, the complexity of binary search algorithm is in the order of $\log_2{(B_{\max})}\log_2{(T^p_{\max})}$.
%
%If we can obtain the global optimal solution of the two subproblems, then it is the global optimal solution of problem \eqref{eq:OptProblem_G}. The proof is omitted due to lack of space.
The original optimization problem is decomposed into $N$ single-user problems. To solve each single-user problem, we search the required bandwidth and optimal prediction horizon in the regions $[0,\overline{B}]$ and $[0,\overline{T}^p]$, respectively, where $\overline{B}$ and $\overline{T}^p$ are the upper bounds of bandwidth and prediction horizon. Therefore, the complexity of the proposed algorithm is around ${\mathcal{O}}\left(N\log_2{(\bar{B})}\log_2{(\bar{T}^p)}\right)$.

The performance loss of the near optimal solution relative to the global optimal solution results from simplification $\varepsilon^{\rm q}_n=\varepsilon^{\rm t}_n$ and the differences between $\hat{\varepsilon}_{n}^{\rm o,\rm min}(B_n)$ and $\varepsilon_{n}^{\rm o,\rm min *}(B_n)$. According to the analysis in \cite{she2018cross} and Proposition \ref{prop:error_bound}, the performance loss is minor. We will further validate the performance loss with numerical results.

\renewcommand{\algorithmicrequire}{\textbf{Input:}}
\renewcommand{\algorithmicensure}{\textbf{Output:}}
\begin{table}[htb]\small
\caption{Algorithm to solve \eqref{eq:OptProblem_G}}
\label{T:algorithm}
\vspace{-0.3cm}
\begin{tabular}{p{16cm}}
\\\hline
\end{tabular}
\vspace{-0.3cm}
\begin{algorithmic}[1]
\REQUIRE User-experienced delay requirement $D_{\rm max}$, reliability requirement $\varepsilon_{\rm max}$, user number $N$, average packet arrival rate $\lambda_n$, each packet duration $D^{\tau}$, slot duration $T_s$, bandwidth of each subcarrier $B_0$, upper bound of bandwidth $\overline{B}$, upper bound of prediction horizon $\overline{T}^p$, transmit power $P^{\rm t}$, user location $d_n$ transition noise $\sigma_j$, initial noise $\sigma_j$, threshold $\delta_j$, $j=1,2,\cdots,F$.
\ENSURE The minimal bandwidth $B_n^{*}$ to ensure URLLC for the $n$th user.
\STATE $B_{\rm L}=B_0,B_{\rm R}=\overline{B}$.
\STATE $B_{\rm b}=\frac{1}{2}\left( B_{\rm L}+B_{\rm R}\right)$.
\STATE Binary search $T^{\rm p}_n$ in a range of $(0,\overline{T}^p]$ and obtain $\hat\varepsilon_{n}^{o,\rm min}(B_{\rm b})$.
\WHILE{$\left|\hat\varepsilon_{n}^{o,\rm min}(B_{\rm b})-\varepsilon_{\rm max}\right|<\varepsilon_{\rm max}$}
\IF{$\hat\varepsilon_{n}^{o,\rm min}(B_{\rm b})>\varepsilon_{\rm max}$}
\STATE $B_{\rm L}=B_{\rm b}$.
\ELSE
\STATE $B_{\rm R}=B_{\rm b}$.
\ENDIF
\STATE $B_{\rm b}=\frac{1}{2}\left( B_{\rm L}+B_{\rm R}\right)$.
\STATE Binary search $T^{\rm p}_n$ in a range of $(0,\overline{T}^p]$ and obtain $\hat\varepsilon_{n}^{o,\rm min}(B_{\rm b})$.
\ENDWHILE
\RETURN $B_{n}^{*}=B_{\rm b}$.
\end{algorithmic}
\vspace{-0.2cm}
\begin{tabular}{p{16cm}}
\\
\hline
\end{tabular}
\vspace{-0.2cm}
\end{table}

%\STATE Set Counter$=0$.
%\FORALL{$i=1,2,\cdots,N_{\rm mc}$}
%\STATE Randomly produce $d_n$ in the range of $[\underline{d},\overline{d}]$.
%\FORALL{$n=1,2,\cdots,N$}
%\STATE $B_n=B_0$. \red{(It seems to be an exhaustive searching method. I suggested using binary search.)}
%\WHILE{$\varepsilon_{n}^{o,\rm min}(B_n)> \varepsilon_{\rm max}$}
%\STATE $B_n=B_n+B_0$
%\STATE Binary search $T^{\rm p}_n$ in a range of $(0,\overline{T}^p]$ and obtain $\varepsilon_{n}^{o,\rm min}(B_n)$. (solve \eqref{eq:OptProblem2_G})
%\ENDWHILE
%\STATE $B^*_n=B_n$.
%\ENDFOR
%\IF{$\sum\limits_{n=1}^N B^*_n>B_{\rm max}$}
%\STATE $\rm Counter=\rm Counter+1$.
%\ENDIF
%\ENDFOR
%\STATE $\rm Availability=\frac{N_{\rm mc}-\rm Counter}{N_{\rm mc}}$.
%\RETURN $N^{\rm c^*}_k = N^{\rm c}_{{k}}{(l-1)}, k=1,...,K$.

\section{Performance Evaluation}
In this section, we evaluate the effectiveness of the proposed co-design method via simulations and experiments.

\subsection{Simulations}
In the simulations, we consider a one-dimensional movement as an example to evaluate the proposed co-design method. With this example, we show how the proposed method helps improving the tradeoffs among latency, reliability and resource utilizations (i.e., bandwidth and antenna). For comparison, the performance achieved by the traditional transmission scheme with no prediction is provided. The simulation parameters are listed in Table \ref{tabel:system_settings}. {In all simulations, SNRs are computed according to $\gamma_n=\frac{a_n g_n P^{\rm t}_n}{\vartheta N_0 B_n}$.} The path loss model is $10\log_{10}({ a}_n) = 35.3 + 37.6 \log_{10}(d_n) + S_n$, where $d_n$ is the distance from the $n$th device to the AP and $S_n$ is the shadowing. The shadowing $S_n$ follows log normal distribution with a zero mean and a standard deviation of $8$. {To ensure the reliability and latency requirements, we consider the worst case of shadowing $S_{\rm w}=-34.1$~dB (i.e., $\Pr\{S_n \le S_{\rm w}\}= 10^{-5}$), which is defined as the probability that the delay and reliability of a device can be satisfied \cite{She2018availability}.}

For the one-dimensional movement, the state transition function in \eqref{eq:state_transition} can be simplified as follows \cite{kay1993fundamentals},
\begin{equation}
\left[ \begin{array}{c}
  r{(k+1)} \\
  v{(k+1)} \\
  a{(k+1)} \\
\end{array} \right]
=
\left[  \begin{array}{ccc}
  1 & T_s & \frac{T_s^2}{2}\\
  0 & 1   & T_s \\
  0 & 0   & 1 \\
\end{array} \right]
\left[ \begin{array}{c}
  r{(k)} \\
  v{(k)} \\
  a{(k)} \\
\end{array} \right]
+
\left[ \begin{array}{c}
  0 \\
  0 \\
  w{(k)} \\
\end{array} \right].\nonumber
\end{equation}
where $r{(k)}$, $v{(k)}$, and $a{(k)}$ represent the location, velocity and acceleration in the $k$th slot, respectively, $w{(k)}$ is the Gaussian noise on acceleration, and $\Phi$ is given by
\begin{equation}\label{eq:Phi}
\Phi =
\left[  \begin{array}{ccc}
  1 & T_s & \frac{T_s^2}{2}\\
  0 & 1   & T_s \\
  0 & 0   & 1 \\
\end{array} \right],
\end{equation}
which follows {Newton's laws of motion}. In predictions, the standard {deviation} of the transition noise of acceleration is $\sigma_w=0.01$~m/s$^2$, and the required threshold is $\delta_l$=$0.1$~m. The standard derivatives of the initial errors of location, velocity and acceleration are set to be $0.01$~m,~$0.2$~m/s, and $0.1$~m/s$^2$, respectively. In practice, the values of initial errors depend on the accuracy of observation and residual filter errors \cite{kay1993fundamentals}.

\begin{table}[!t]
\renewcommand{\arraystretch}{1}
\caption{Simulation Parameters \cite{3GPP2017Scenarios}}
\label{tabel:system_settings}
\centering
\begin{tabular}{|l|l|}
\hline
\textbf{Parameters}                                      & \textbf{Values}        \\
\hline
Maximal transmit power of a user $P_{\rm t}$          & $23$~dBm          \\
\hline
Single-sided noise spectral density $N_0$       & $-174$~dBm/Hz       \\
\hline
Information load per block b                             & $160$~bits         \\
\hline
Average packet arrival rate $\lambda$                             & $100$~packets/second         \\
\hline
Slot duration $T_{\rm s}$                             & $0.1$~ms         \\
\hline
Transmission duration $D_{\tau}$                                     & $0.5$~ms  \\
\hline
Delay of core network and backhaul $D_{\rm r}$                    & $10$~ms   \\
\hline
\end{tabular}
\end{table}

\subsubsection{Single-user scenarios} In single-user scenarios, the distance between the user and the AP is set to be $200$~m. To evaluate the proposed co-design method, the prediction horizon $T^{\rm p}_n$ is optimized to obtain the minimal overall error probability.

\begin{figure}[htb]
	%\vspace{-0.2cm}
	\centering
	\begin{minipage}[t]{0.65\textwidth}
		\includegraphics[width=1\textwidth]{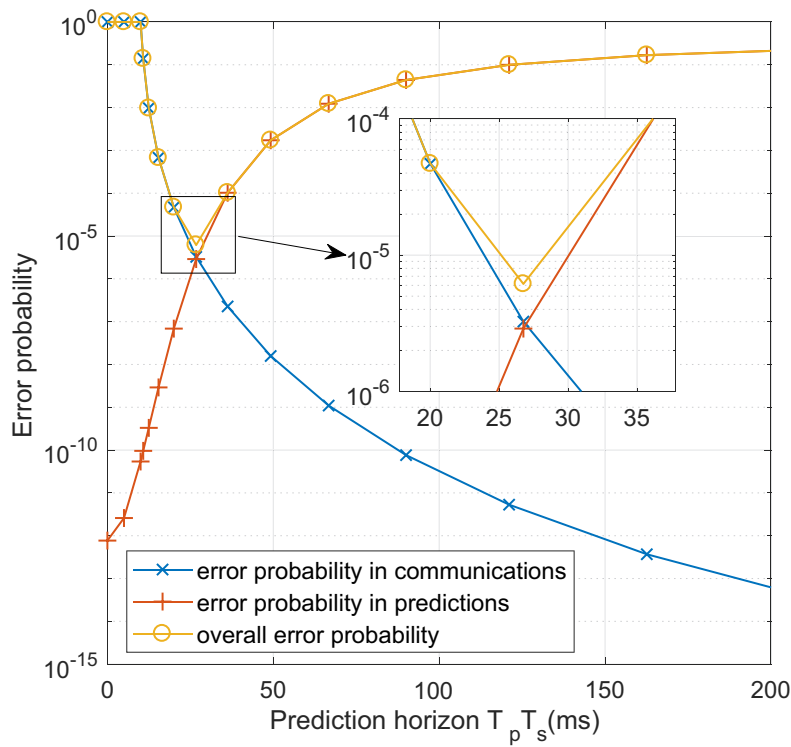}
	\end{minipage}
	%\vspace{-0.1cm}
	\caption{Joint optimization of predictions and communications: the packet loss probability $\varepsilon_c$ in communications, the prediction error probability $\varepsilon_p$, and the over error probability $\varepsilon_o$ are drawn as functions of prediction horizon $T_{\rm p}T_{\rm s}$.  }
	\label{fig:co_disign}
	%\vspace{-0.3cm}
\end{figure}

Under the given delay requirement (i.e., $D_{\max} = 0$~ms), the packet loss probability in communications $\varepsilon^{\rm c}_n$, the prediction error probability $\varepsilon^{\rm p}_n$, and the overall error probability $\varepsilon^{\rm o}_n$ are shown in Fig. \ref{fig:co_disign}. To achieve target reliability, the bandwidth $B$ is set as $B=440$~KHz and the number of antennas at the AP is set to be $N_{\rm r}=32$. It should be noted that the reliability depends on the amount of bandwidth and the number of antennas, but the trend of the overall reliability does not change.

In Fig. \ref{fig:co_disign}, the communication delay and prediction horizon are set to be equal, i.e., $D^{\rm q}_n+D^{\rm t}_n=T^{\rm p}_n$. In this case, user experienced delay is zero. The results in Fig. \ref{fig:co_disign} show that when the E2E communication delay $D^{\rm q}_n+D^{\rm t}_n=T^{\rm p}_n<10$~ms, i.e., less than the delays in the core network and the backhaul $D^{\rm r}_n$, it is impossible to achieve zero latency without prediction. When $D^{\rm q}_n+D^{\rm t}_n=T^{\rm p}_n>10$~ms, the required transmission duration $K D^{\tau}_n$ increases with prediction horizon $T^{\rm p}_n$. As a result, the overall error probability, $\varepsilon^{\rm o}_n$, is first dominated by $\varepsilon^{\rm c}_n$ and then by $\varepsilon^{\rm p}_n$. As such, $\varepsilon^{\rm o}_n$ first decreases and then increases with $T^{\rm p}_n$. {The results in Fig. \ref{fig:co_disign} indicate that the reliability achieved by the proposed method is $6.52 \times 10^{-6}$ with $T^{\rm p}_n=26.8$~ms, $D^{\rm t*}_n=12.5$~ms, $D^{\rm q*}_n=14.3$~ms and $K_n^*=5$. The optimal solution obtained by exhaustive search is $6.15 \times 10^{-6}$. The gap between above two solutions is $3.7 \times 10^{-7}$, which is very small.}

\begin{figure}[htb]
	%\vspace{-0.2cm}
	\centering
	\begin{minipage}[t]{0.65\textwidth}
		\includegraphics[width=1\textwidth]{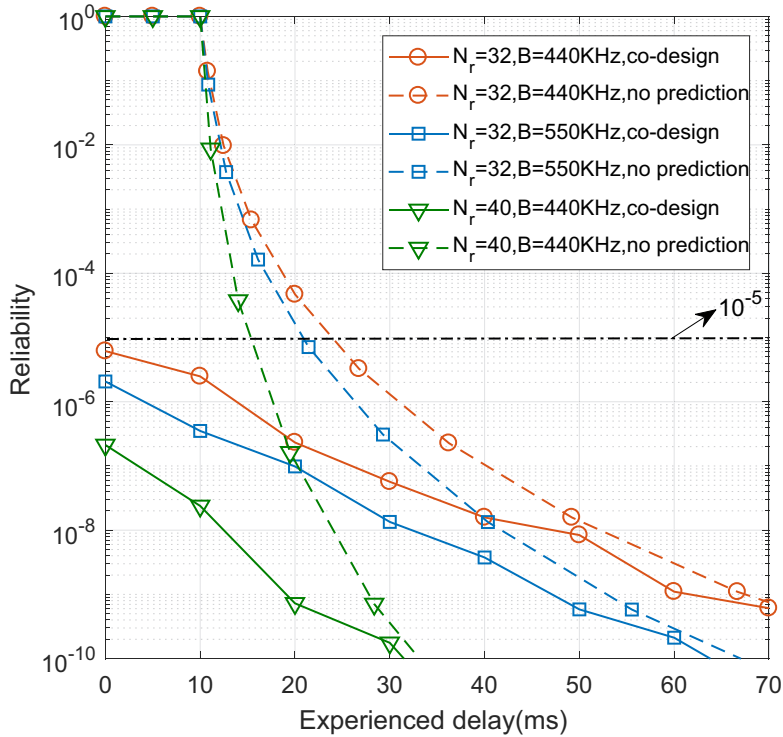}
	\end{minipage}
	%\vspace{-0.1cm}
	\caption{Comparison of reliability-delay tradeoff curves between co-design and no predictions with different bandwidth $B$ and numbers of received antennas $N_{\rm r}$.}
	\label{fig:comp_tradeoff}
	%\vspace{-0.3cm}
\end{figure}

In Fig. \ref{fig:comp_tradeoff}, the proposed co-design method is compared with a baseline method without prediction. When there is no prediction, the user experienced delay equals to communication delay. The results in Fig. \ref{fig:comp_tradeoff} show that when the requirement on user experienced delay is less than $10$~ms, it cannot be satisfied without prediction. When the required user experienced delay is larger than $10$~ms, the reliability achieved by the co-design method is much better than the baseline method. In other words, by prediction and communication co-design, the tradeoff between user experienced delay and overall reliability can be improved remarkably. Particularly, in the case $N_{\rm r}=32$ and $B=440$~KHz, to ensure the same reliability $10^{-5}$, the user experienced delay can be reduced by $23$~ms and zero-latency can be achieved by the proposed co-design method.

\subsubsection{Multiple-user scenarios}
{In multiple-user scenarios, we will consider two scenarios: the distribution of large-scale fading of
the mobile devices is available/unavaibale.} In the first scenario, the distances from devices to the AP are uniformly distributed in the region $[50,200]$~m. In the second scenario, the worst case is considered in the optimization, i.e., the distances from all the devices to the AP are $200$~m.

\begin{figure}[htb]
	%\vspace{-0.2cm}
	\centering
	\begin{minipage}[t]{0.65\textwidth}
		\includegraphics[width=1\textwidth]{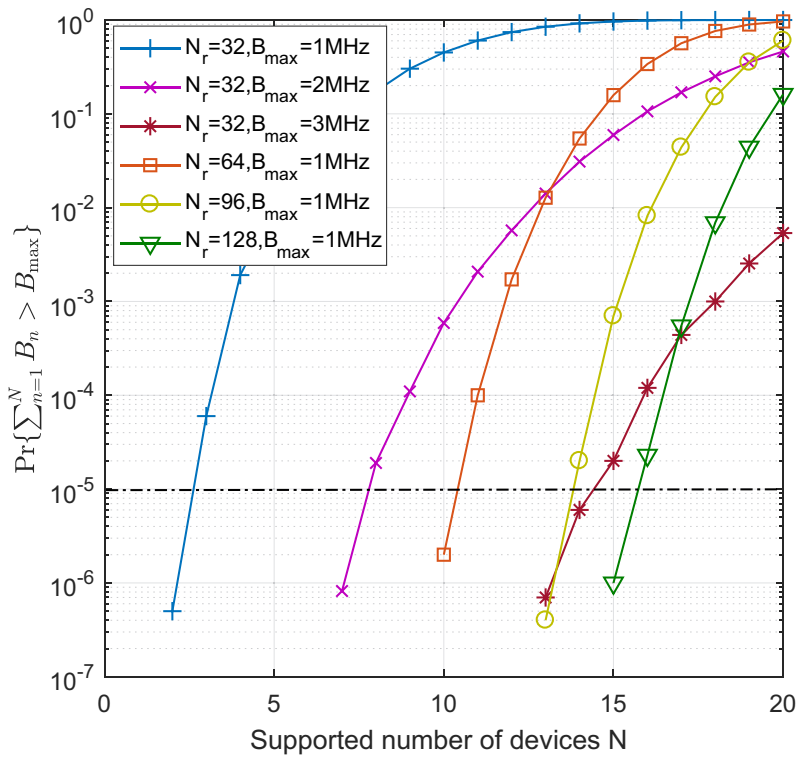}
	\end{minipage}
	%\vspace{-0.1cm}
	\caption{$\Pr\{\sum_{n=1}^N{ B_n} > B_{\max}\}$ v.s. the number of devices when the distribution of {large-scale fading of devices} is known.}
	\label{fig:N_knowLoc}
	%\vspace{-0.3cm}
\end{figure}

Since the {large-scale fading} of devices are random variables in the first scenario, the sum of the required bandwidth is also a random variable. In Fig. \ref{fig:N_knowLoc}, we illustrated the probability that the sum of the required bandwidth is smaller than $B_{\max}$. For URLLC services, we need to guarantee the delay and reliability requirements with high probability, e.g., $99.999$~\%. The results in Fig. \ref{fig:N_knowLoc} show that when $B_{\max} = 1$~MHz and $N_{\rm r} = 32$, the system can only support $2$ devices.  By doubling the number of antennas (or the total bandwidth), $10$ (or $7$) devices can be supported. This implies that increasing the number of antennas at the AP is an efficient way to increase the number of devices that can be supported by the system. {This is because SNR increases with the number of antennas due to array gain. To achieve the same reliability, i.e., packet loss probability, higher order modulation schemes can be used if more antennas are deployed at the AP. Since the spectrum efficiency increases with the order of the modulation scheme, more URLLC devices can be supported with a given amount of bandwidth.}

\begin{figure}[htb]
	%\vspace{-0.2cm}
	\centering
	\begin{minipage}[t]{0.65\textwidth}
		\includegraphics[width=1\textwidth]{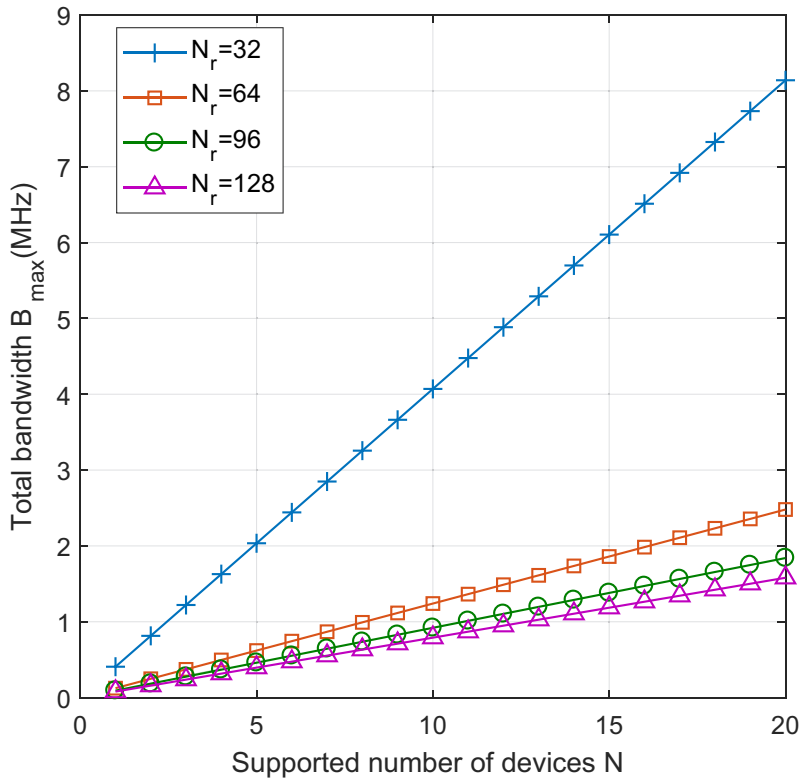}
	\end{minipage}
	%\vspace{-0.1cm}
	\caption{Required total bandwidth v.s. number of devices when the distribution of {large-scale fading of devices} is unknown.}
	\label{fig:N_unknowLoc}
	%\vspace{-0.3cm}
\end{figure}

If the distribution of {large-scale fading} of devices is unknown, the worst case is considered. Then, the total bandwidth that is required to support a given number of devices is deterministic. The results in Fig. \ref{fig:N_unknowLoc} show that the required total bandwidth linearly increases with the number of devices. This is because the required bandwidth for different devices are the same since the worst case is considered for all the devices. In addition, by increasing the number of antennas from $32$ to $64$, we can save $75$~\% of bandwidth. This implies that increasing the number of antennas is an efficient way to improve spectrum efficiency of URLLC.

\subsection{Experiments}

\begin{figure}[!t]
	%\vspace{-0.2cm}
	\centering
	\begin{minipage}[t]{0.65\textwidth}
		\includegraphics[width=1\textwidth]{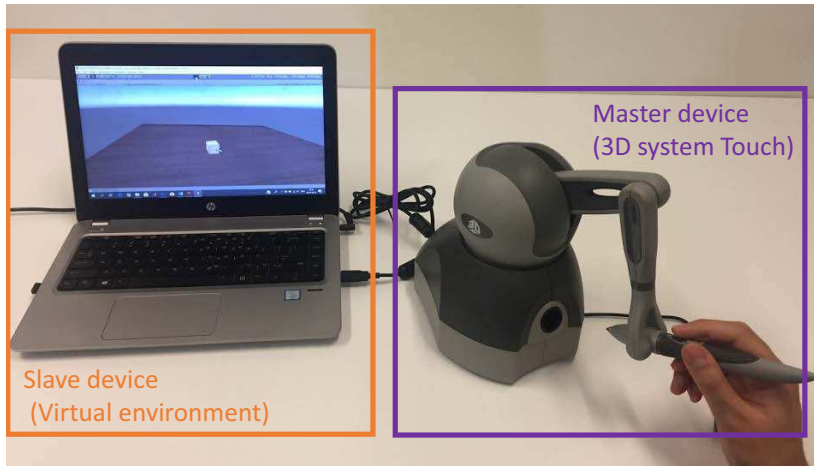}
	\end{minipage}
	%\vspace{-0.1cm}
	\caption{Experiment to obtain real movement data in Tactile Internet.}
	\label{fig:experimemnt_setting}
	%\vspace{-0.3cm}
\end{figure}

To validate whether mobility prediction works for URLLC in practice, we record the real movement data from the experiment shown in Fig. \ref{fig:experimemnt_setting}. In this experiment, a typical application of Tactile Internet is implemented in a virtual environment, where a box of hazardous chemicals or radioactive substances is dragged to move on the floor by a virtual slave device. A tactile hardware device named 3D System Touch (previously named Phantom Omni, or Geomagic) is used as a master device, which sends real time location information to the virtual slave device. A cable is used to connect the master device to the virtual slave device in a virtual environment. The slave device in the virtual environment receives the locations from the master device, so it can move synchronously with the master device.

Human operators are invited to drag the virtual box from one corner to another corner of the floor in the virtual environment. In this experiment, we mainly interested in the motion prediction, so the location information on the x-axis produced by the tactile hardware device is recorded and used to verify the predictions. A general linear prediction method in \eqref{eq:n_step_prediction_model} is used to predict the future state system. Since only location information is available from the hardware, the velocity and acceleration are obtained from the first and the second order differences of locations \cite{mathews2004numerical}. Moreover, due to the limitation of the hardware, the duration of each slot is $T_{\rm s}=1$~ms.

{The prediction error probabilities, $\varepsilon^{\rm p}_n$, with different thresholds, $\delta$, are shown in Table \ref{table:Prediction_error2}, where the prediction horizon, $n T_{\rm s}$, is fixed. The results in Table \ref{table:Prediction_error2} show that for the constant prediction horizon $n T_{\rm s}=5$~ms or $n T_{\rm s}=20$~ms, $\varepsilon^{\rm p}_n$ decreases with the required threshold $\delta$.}

{The relation between the prediction error probability and the prediction horizon is shown in Table \ref{table:Prediction_error1}, where the required threshold is fixed. The results indicate that $\varepsilon^{\rm p}_n$ increases with $n T_{\rm s}$. This observation consists with Lemma \ref{prop:mono_epsilon_p}.}

The results in Tables \ref{table:Prediction_error2} and \ref{table:Prediction_error1} imply that prediction and communication co-design has the potential to achieve zero-latency in practice. It should be noted that the results from the experiment are generally worse than those of the simulations. This is because we only have the location information of the device, and extra estimation errors are introduced during the estimations of the velocity and acceleration.

\begin{table}[!t]
\renewcommand{\arraystretch}{1}
\caption{{Prediction Error Probability with fixed $n T_{\rm s}$}}
\label{table:Prediction_error2}
\centering
{\begin{tabular}{|c|c|c|}
\hline
$\delta$(m)   &  $\varepsilon^{\rm p}_n$ ($n T_{\rm s}=5~ms$) &  $\varepsilon^{\rm p}_n$ ($n T_{\rm s}=20~ms$)  \\
\hline
$0.002$         &   $2.95\times 10^{-4}$ & 0.45  \\
\hline
$0.01$         &   $1.62\times 10^{-5}$  & 0.42 \\
\hline
$0.02$        &    $6.61\times 10^{-6}$  & $3.2\times 10^{-3}$ \\
\hline
$0.1$         &   $2.40\times 10^{-6}$   & $2.40\times 10^{-5}$ \\
\hline
$0.2$        &   $1.80\times 10^{-6}$    & $7.82\times 10^{-6}$ \\
\hline
\end{tabular}}
\end{table}

\begin{table}[!t]
\renewcommand{\arraystretch}{1}
\caption{{Prediction Error Probability with given $\delta$}}
\label{table:Prediction_error1}
\centering
{\begin{tabular}{|c|c||c|c|}
\hline
$n T_{\rm s}$(ms)   &  $\varepsilon^{\rm p}_n$ ($\delta=0.002~m$)    &  $n T_{\rm s}$(ms)   &  $\varepsilon^{\rm p}_n$ ($\delta=0.2~m$)    \\
\hline
$1$         &   $3.00\times 10^{-6}$ & $10$         &   $1.80\times 10^{-6}$       \\
\hline
$2$         &   $1.62\times 10^{-5}$ & $20$         &   $7.82\times 10^{-6}$ \\
\hline
$3$        &    $3.49\times 10^{-5}$ & $30$        &    $2.71\times 10^{-5}$  \\
\hline
$4$         &   $5.77\times 10^{-5}$ & $40$         &   $4.45\times 10^{-5}$\\
\hline
$5$        &   $2.95\times 10^{-4}$  & $50$         &   $1.69\times 10^{-4}$\\
\hline
\end{tabular}}
\end{table}

%\begin{table}[!t]
%\renewcommand{\arraystretch}{1}
%\caption{Prediction Error Probability}
%\label{tabel:Prediction_error}
%\centering
%\begin{tabular}{|c|c|c|}
%\hline
%$n T_{\rm s}$(ms)   & $\delta$ (m)  &  $\varepsilon^{\rm p}_n$    \\
%\hline
%$1$       & $0.002$  &   $3.00\times 10^{-6}$       \\
%\hline
%$2$       & $0.002$  &   $1.62\times 10^{-5}$ \\
%\hline
%$5$      & $0.002$  &    $2.95\times 10^{-4}$  \\
%\hline
%$5$       & $0.005$  &   $3.73\times 10^{-5}$\\
%\hline
%$5$       & $0.01$ &   $1.62\times 10^{-5}$ \\
%\hline
%$10$       & $0.01$ &  $6.55\times 10^{-5}$ \\
%\hline
%$30$       & $0.1$ &  $5.17\times 10^{-5}$ \\
%\hline
%\end{tabular}
%\end{table}

\section{Conclusions}
In this paper, we studied how to achieve URLLC by prediction and communication co-design. We first derived the decoding error probability, the queueing delay violation probability, and the prediction error probability in closed-form expressions. Then, we established an optimization framework for maximizing the number of devices that can be supported in a system by optimizing time and frequency resources in the communication system and the prediction horizon in the prediction system. Simulation results show that by prediction and communication co-design the tradeoff between delay and reliability can be improved remarkably{, or we can improve the spectrum efficiency subject to the delay and reliability constraints}. In addition, an experiment was carried out to validate the accuracy of prediction in a remote-control system. The results showed that the proposed concept on prediction and communication co-design works well in the practical remote-control system.

\appendices

\section{Proof of {Lemma} \ref{prop:mono_epsilon_p}}
\label{appendix:proposition:monop}
\begin{proof}
To prove this lemma, we need to prove that for any $T^{p,1}_n<T^{p,2}_n$, $\varepsilon^{\rm p}_n(T^{p,1}_n)<\varepsilon^{\rm p}_n(T^{p,2}_n)$ holds. From \eqref{eq:sigma_n_j}, we have
\begin{equation*}\label{eq:diff_sigma}
  \sigma_{j}^2(T^p_n+1)-\sigma_{j}^2(T^p_n)=\sum\limits_{m=1}^{F}\phi_{n,j,m,n}\sigma_{m}^2>0.
\end{equation*}
As such, we can conclude that $\sigma_{j}(T^p_n), j=1,2,\cdots,N$, increases with $T^p_n$.

Moreover, from \eqref{eq:PredictErrCalc}, we can see that $\varepsilon^{\rm p}_n$ increases with $\sigma_{j}(T^p_n), j=1,2,\cdots,N$. Therefore, $\varepsilon^{\rm p}_n$ increases with the prediction horizon $T^p_n$. This completes the proof.
\end{proof}

\section{Derivations of \eqref{eq:varepsilon_q} and \eqref{eq:varepsilon_q2}}
\label{appendix:lambert}

The equation \eqref{eq:EB} can be re-expressed as
\begin{equation}\label{eq:Lambert1}
  \frac{\ln(1/\varepsilon^{\rm q}_n)+\lambda_n D^{\rm q}_n}{\lambda_n D^{\rm q}_n}=\exp\left[ \frac{\ln(1/\varepsilon^{\rm q}_n)+\lambda_n D^{\rm q}_n}{D^{\rm q}_n E^{\rm B}_n} -\frac{\lambda_n}{E^{\rm B}_n} \right],
\end{equation}
and
\begin{equation}\label{eq:Lambert2}
  -\frac{\lambda_n}{E^{\rm B}_n} \exp \left( -\frac{\lambda_n}{E^{\rm B}_n} \right)=  \frac{\ln(1/\varepsilon^{\rm q}_n)+\lambda_n D^{\rm q}_n}{-D^{\rm q}_n E^{\rm B}_n} \exp\left[ \frac{\ln(1/\varepsilon^{\rm q}_n)+\lambda_n D^{\rm q}_n}{-D^{\rm q}_n E^{\rm B}_n} \right].
\end{equation}
According to the definition of Lambert function, \eqref{eq:Lambert2} can be written as
\begin{equation}
  \frac{\ln(1/\varepsilon^{\rm q}_n)+\lambda_n D^{\rm q}_n}{-D^{\rm q}_n E^{\rm B}_n}=\mathbb{W}\left[ -\frac{\lambda_n}{E^{\rm B}_n} \exp \left( -\frac{\lambda_n}{E^{\rm B}_n} \right) \right].
\end{equation}
It should be noted that when $-\frac{\lambda_n}{E^{\rm B}_n} \exp \left( -\frac{\lambda_n}{E^{\rm B}_n} \right)<0$, the Lambert function has two branches according to the range of $\frac{\ln(1/\varepsilon^{\rm q}_n)+\lambda_n D^{\rm q}_n}{D^{\rm q}_n E^{\rm B}_n}$. Specifically, we have
\begin{equation}\label{eq:Lambert3}
  \frac{\ln(1/\varepsilon^{\rm q}_n)+\lambda_n D^{\rm q}_n}{-D^{\rm q}_n E^{\rm B}_n}=\left\{
  \begin{aligned}
  &\mathbb{W}_{0} \left[ -\frac{\lambda_n}{E^{\rm B}_n} \exp \left( -\frac{\lambda_n}{E^{\rm B}_n} \right) \right],\ -1< \frac{\ln(1/\varepsilon^{\rm q}_n)+\lambda_n D^{\rm q}_n}{-D^{\rm q}_n E^{\rm B}_n} <0\\
  &\mathbb{W}_{-1} \left[ -\frac{\lambda_n}{E^{\rm B}_n} \exp \left( -\frac{\lambda_n}{E^{\rm B}_n} \right) \right],\ \frac{\ln(1/\varepsilon^{\rm q}_n)+\lambda_n D^{\rm q}_n}{-D^{\rm q}_n E^{\rm B}_n} \ge 0
  \end{aligned}
  \right.
\end{equation}
In the first case in \eqref{eq:Lambert3}, $\mathbb{W}_{0} \left[ -\frac{\lambda_n}{E^{\rm B}_n} \exp \left( -\frac{\lambda_n}{E^{\rm B}_n} \right) \right] = -\frac{\lambda_n}{E^{\rm B}_n} = \frac{\ln(1/\varepsilon^{\rm q}_n)+\lambda_n D^{\rm q}_n}{-D^{\rm q}_n E^{\rm B}_n}$. We can obtain that $\varepsilon_n^{\rm q}=1$, which does not satisfy the reliability requirement. Thus, the first case in \eqref{eq:Lambert3} can be removed. As such, we have
\begin{equation}\label{Lambert4}
  \frac{\ln(1/\varepsilon^{\rm q}_n)+\lambda_n D^{\rm q}_n}{-D^{\rm q}_n E^{\rm B}_n}=\mathbb{W}_{-1} \left[ -\frac{\lambda_n}{E^{\rm B}_n} \exp \left( -\frac{\lambda_n}{E^{\rm B}_n} \right) \right],
\end{equation}
and
\begin{equation}\label{Lambert5}
  \varepsilon^{\rm q}_n=\exp\left\{ D_n^{\rm q}  E_n^{\rm B} \mathbb{W}_{-1} \left[ -\frac{\lambda_n}{E^{\rm B}_n} \exp \left( -\frac{\lambda_n}{E^{\rm B}_n} \right) \right] + D_n^{\rm q} \lambda_n \right\}.
\end{equation}

\section{Proof of {Lemma} \ref{prop:mono_epsilon_q}}
\label{appendix:proposition:monoq}

\begin{proof}
According to \eqref{eq:varepsilon_q}, we have
\begin{equation}\label{eq:vare_q1}
  \ln{(\varepsilon^{\rm q}_n)}=D^{\rm q}_n\phi(\lambda_n,E^{\rm B}_n).
\end{equation}
Since $\varepsilon^{\rm q}_n$ is in the order of $10^{-5}$ to $10^{-8}$ and $D^{\rm q}_n>0$, $\ln{(\varepsilon^{\rm q}_n)}<0$, and thus $\phi(\lambda_n,E^{\rm B}_n)<0$. As such, $\varepsilon^{\rm q}_n$ decreases with $D^{\rm q}_n$ in \eqref{eq:varepsilon_q} when $\phi(\lambda_n,E^{\rm B}_n)$ is given. The proof follows.

\end{proof}

\section{Proof of Proposition \ref{prop:mono_n}}
\label{appendix:proposition:mono_n}

\begin{proof}
According to \eqref{eq:OptRelationEqual_G}, we have $D^{\rm q}_n+D^{\rm t}_n= D_{\rm max}+T^p_n-D^{\rm r}_n$. To prove this proposition, we need to prove that $\varepsilon^{\rm q}_n$ or $\varepsilon^{\rm t}_n$ decreases with $D^{\rm q}_n+D^{\rm t}_n$.

Next, we will prove $D^{\rm q}_n$ increases with $D^{\rm t}_n$, and thus $D^{\rm q}_n+D^{\rm t}_n$ increases with $D^{\rm t}_n$. According to \eqref{eq:varepsilon_q} and \eqref{eq:varepsilon_q2}, we have
\begin{equation*}
  D^{\rm q}_n=\frac{\ln{(\varepsilon^{\rm q}_n)}}{\phi{(\lambda_n,E^{\rm B}_n)}},
\end{equation*}
where
\begin{equation*}
  \phi(\lambda_n,E^{\rm B}_n)= \frac{\mathbb{W}_{-1}\left(-\lambda_n D^{\rm t}_n e^{-\lambda_n D^{\rm t}_n}\right)}{D^{\rm t}_n} + \lambda_n.
\end{equation*}
To check the monotonicity of $D^{\rm q}_n$ in terms of $\varepsilon^{\rm q}_n$ and $D^{\rm t}_n$, we have the following partial derivatives,
\begin{equation}\label{eq:pd1}
  \frac{\partial{D^{\rm q}_n}}{\partial{\varepsilon^{\rm q}_n}}=\frac{1}{\varepsilon^{\rm q}_n\phi{(\lambda_n,E^{\rm B}_n)}} < 0,
\end{equation}
and
\begin{equation}\label{eq:pd2}
  \frac{\partial{D^{\rm q}_n}}{\partial{D^{\rm t}_n}} = \frac{\ln{(\varepsilon^{\rm q}_n)} \mathbb{W}_{-1}(-\lambda_n D^{\rm t}_n e^{-\lambda_n D^{\rm t}_n})}{\left[ \mathbb{W}_{-1}(-\lambda_n D^{\rm t}_n e^{-\lambda_n D^{\rm t}_n})+1 \right] \left[ \mathbb{W}_{-1}(-\lambda_n D^{\rm t}_n e^{-\lambda_n D^{\rm t}_n})+\lambda_n D^{\rm t}_n \right] } >0.
\end{equation}
As such, we prove $D^{\rm q}_n$ increases with $D^{\rm t}_n$ when $\varepsilon^{\rm q}_n$ is given. According to {Lemma} \ref{prop:mono_epsilon_t}, $\varepsilon^{\rm t}_n$ strictly decreases with the transmission delay $D^{\rm t}_n$. Since $\varepsilon^{\rm q}_n=\varepsilon^{\rm t}_n$, $\varepsilon^{\rm q}_n$ also strictly decreases with the transmission delay $D^{\rm t}_n$. According to \eqref{eq:pd1}, $D^{\rm q}_n$ increases with a smaller $\varepsilon^{\rm q}_n$. So $D^{\rm q}_n$ increases with $D^{\rm t}_n$ when $\varepsilon^{\rm q}_n$ is determined by $D^{\rm t}_n$.

In summary, $\varepsilon^{\rm q}_n$ or $\varepsilon^{\rm t}_n$ decreases with $D^{\rm t}_n$ and $D^{\rm q}_n+D^{\rm t}_n$, and thus decreases with $T^p_n$. This completes the proof.

\end{proof}

\section{Proof of Proposition \ref{prop:error_bound}}
\label{appendix:proposition:bound}
\begin{proof}
In this appendix, we use the notation $\varepsilon^{\rm o}_n(T^p_n,B_n)$, (or $\varepsilon^{\rm t}_n(T^p_n,B_n)$ or $\varepsilon^{\rm p}_n(T^p_n,B_n)$) to represent the relationship between the prediction horizon and the overall reliability (or the decoding error probability or the prediction error probability). For notational simplicity, we first omit $B_n$.

To prove this proposition, we first introduce an upper bound of $\varepsilon^{\rm o}_n(T^p_n) = 2\varepsilon^{\rm t}_n(T^p_n + D_{\max}) + \varepsilon^{\rm p}_n(T^p_n)$, i.e., $\varepsilon_{\rm{o},n}^{\rm ub}(T^p_n) = 2\max\{2\varepsilon^{\rm t}_n(T^p_n + D_{\max}) , \varepsilon^{\rm p}_n(T^p_n)\}$.

Suppose $\tilde{T}^p_n$ is the maximal prediction horizon that satisfies $2\varepsilon^{\rm t}_n(T^p_n + D_{\max}) - \varepsilon^{\rm p}_n(T^p_n) > 0$ for all $0 \le T^p_n \le \tilde{T}^p_n$, and hence $\varepsilon_n^{\rm o,\rm ub}(T^p_n) = 4\varepsilon^{\rm t}_n(T^p_n + D_{\max})$, which strictly decreases with $T^p_n$. On the other hand, when $T^p_n > \tilde{T}^p_n$, $2\varepsilon^{\rm t}_n(T^p_n + D_{\max}) - \varepsilon^{\rm p}_n(T^p_n) < 0$, and hence $\varepsilon_n^{\rm o,\rm ub}(T^p_n) = 2\varepsilon^{\rm p}_n(T^p_n)$, which strictly increases with $T^p_n$. In other words, $\varepsilon_n^{\rm o,\rm ub}(T^p_n)$ strictly decreases with $T^p_n$ when $T^p_n \le \tilde{T}^p_n$ and strictly increases with $T^p_n$ when $T^p_n > \tilde{T}^p_n$. Therefore, the upper bound $\varepsilon_n^{\rm o,\rm ub}(T^p_n)$ is minimized at $\hat{T}^p_n = \tilde{T}^p_n$ or $\hat{T}^p_n = \tilde{T}^p_n+1$.

Let $2\varepsilon^{\rm t}_n(\hat{T}^p_n + D_{\max}) - \varepsilon^{\rm p}_n(\hat{T}^p_n) = \Delta$, where $\Delta$ is the small gap between $2\varepsilon^{\rm t}_n$ and $\varepsilon^{\rm p}_n$ at $\hat{T}^p_n$, which is every closed to zero. We have
\begin{equation}\label{eq:Step1}
\varepsilon^{\rm o}_n(\hat{T}^p_n) \approx \varepsilon_n^{\rm o,\rm ub}(\hat{T}^p_n),
\end{equation}
Besides, $\varepsilon_n^{\rm o,\rm ub}(\hat{T}^p_n)$ is the minimum of $\varepsilon_n^{\rm o,\rm ub}(T^p_n), \forall n\in [0,\infty)$, and hence
\begin{align}
\varepsilon_n^{\rm o,\rm ub}(\hat{T}^p_n) \leq \varepsilon_n^{\rm o,\rm ub}(T_n^{p*}), \label{eq:Step2}
\end{align}
where $T_n^{p*}$ is the optimal prediction horizon that minimizes $\varepsilon^{\rm o}_n(T^p_n)$. According to the definition of $\varepsilon_n^{\rm o,\rm ub}(T^p_n)$, we have
\begin{align}
\varepsilon_n^{\rm o,\rm ub}(T_n^{p*})&= 2 \max\{2\varepsilon^{\rm t}_n(T_n^{p*} + D_{\max}) , \varepsilon^{\rm p}_n(T_n^{p*})\} \nonumber\\
  & < 2 \left[2\varepsilon^{\rm t}_n(T_n^{p*} + D_{\max}) + \varepsilon^{\rm p}_n(T_n^{p*})\right] \nonumber\\
  & = 2 \varepsilon^{\rm o}_n(T_n^{p*})\label{eq:Step3}.
\end{align}
From \eqref{eq:Step1}, \eqref{eq:Step2} and \eqref{eq:Step3}, we have $\varepsilon^{\rm o}_n(\hat{T}^p_n) < 2 \varepsilon^{\rm o}_n(T_n^{p*})$, i.e., $\varepsilon^{\rm o}_n(\hat{T}^p_n) - \varepsilon^{\rm o}_n(T_n^{p*}) < \varepsilon^{\rm o}_n(T_n^{p*})$.

Since $\varepsilon^{\rm o}_n(\hat{T}^p_n)$ and $\varepsilon^{\rm{o}}_n(T_n^{p*})$ are defined as $\hat{\varepsilon}_{n}^{\rm{o},\rm min}(\hat{T}^p_n,B_n)$ and ${\varepsilon}_{\rm{o},n}^{\rm min *}(T_n^{p*},B_n)$, respectively. So we have $\hat{\varepsilon}_{n}^{\rm{o},\rm min}(\hat{T}^p_n,B_n)-{\varepsilon}_{n}^{\rm{o},\rm min *}(T_n^{p*},B_n)<{\varepsilon}_{n}^{\rm{o},\rm min *}(T_n^{p*},B_n)$. The proof follows.
\end{proof}

\ifCLASSOPTIONcaptionsoff
  \newpage
\fi

\bibliographystyle{IEEEtran}
\bibliography{main}

\end{document}